\documentclass[aps,prx,twocolumn,floatfix,superscriptaddress,graphics,longbibliography]{revtex4-1}

\usepackage{amsmath,amssymb,graphicx,color,braket,makeidx,subfigure, dsfont}
\usepackage{bm}
\usepackage[retainorgcmds]{IEEEtrantools}
\usepackage{graphicx}
\usepackage{mathrsfs}
\usepackage{amsmath}
\usepackage{amssymb}
\usepackage{color}
\usepackage{amsfonts}
\usepackage{amsthm}
\usepackage{times,txfonts}
\usepackage{nicefrac}
\usepackage[colorlinks=true,linkcolor=blue,urlcolor=blue,citecolor=blue,pdfusetitle]{hyperref}
\usepackage{framed}
\usepackage{yfonts}
\usepackage{cancel}
\usepackage[dvipsnames]{xcolor}
\usepackage[normalem]{ulem}

\newtheorem{theorem}{Theorem}

\newtheorem*{definition*}{Def}
\newtheorem*{theorem*}{Theorem}
\newtheorem*{corollary*}{Corollary}
\newtheorem*{lemma*}{Lemma}
\newtheorem*{teorema*}{Teorema}

\newcommand{\ud}{\mathrm{d}}

\newcommand{\gtl}[1]{{\color{black}#1}}
\newcommand{\GGrev}[1]{\textcolor{black}{#1}}

\newcommand{\stkout}[1]{\ifmmode\text{\sout{\ensuremath{#1}}}\else\sout{#1}\fi}
\newcommand{\fix}[2]{{\color{black}#2}}

\begin{document}

\title{Quantum thermoelectric transmission functions with minimal current fluctuations}
\date{\today}
\author{Andr\'e M. Timpanaro}
\affiliation{Universidade Federal do ABC,  09210-580 Santo Andr\'e, Brazil}
\author{Giacomo Guarnieri}
\affiliation{Department of Physics, University of Pavia, Via Bassi 6, 27100, Pavia, Italy}
\author{Gabriel T. Landi}
\affiliation{Instituto de F\'isica da Universidade de S\~ao Paulo,  05314-970 S\~ao Paulo, Brazil.}

\begin{abstract}
Thermodynamic Uncertainty Relations (TURs) \GGrev{represent a benchmark result in non-equilibrium physics that allows to place fundamental} lower bounds on the noise-to-signal ratio (precision) of currents in nanoscale devices.
Originally formulated for classical time-homogeneous Markov processes, these relations, were shown to be violated in \GGrev{thermoelectric engines and photovoltaic devices supporting quantum-coherent transport}.
However, the extent to which these violations may occur still represent\fix{}{s} a missing piece of the puzzle.
In this work we provide \GGrev{such answer in a definitive way within the general Landauer-B\"{u}ttiker formalism~\fix{}{for non-interacting systems,} beyond any perturbative regime, e.g. linear response.
In particular, using analytical constrained-optimization techniques, we rigorously demonstrate that the transmission function which maximizes the reliability of thermoelectric devices (i.e. which minimizes the fluctuations of its steady-state currents) for fixed average power and efficiency is a collection of boxcar functions.}
This allows us to show that TURs can be violated by arbitrarily large amounts, depending on the temperature and chemical potential gradients, thus providing guidelines to the design of optimal devices.

\end{abstract}

\maketitle

\section{Introduction}

Since the industrial revolution, thermal machines represent a crucial component of our technological landscape, and optimizing their performance has been an overarching goal of engineers and scientists for centuries.
At the macroscopic level, the Second law of Thermodynamics imposes a fundamental trade-off between the average power output of thermal machines and their efficiency. Specifically, the Carnot bound provides the maximum efficiency that can be reached, which is achieved only for infinitely slow processes, thus  entail vanishing dissipation but also no output power ~\cite{Carnot1824}. Conversely, finite power output implies a lower efficiency than Carnot’s.

The 20th century saw a relentless march towards miniaturization, leading to the development and control of nano-scale devices that are able to interconvert heat and particles at the microscopic scale.  Thermoelectric devices represent one of the most prominent examples of this trend. These devices have the ability to convert heat into electricity, and vice versa, with a range of applications including the generation of renewable energy and the cooling of electronic devices ~\cite{Beretta2019, Benenti2017}.
According to the Landauer-Buttiker theory, in the non-linear regime, thermoelectric devices can be fully characterized by means of the transmission function ~\cite{Landauer1957}. Over the years, significant research has been directed towards optimizing these devices to reach their best performance. 
A first paper by Mahan and Sofo ~\cite{Mahan1996} showed that a narrowly-peaked distribution of the  transmission function would yield the highest efficiency in linear response regime. In accordance with the second law’s predictions, the corresponding power output, however, would vanish. This result was later generalized in Ref.~\cite{Esposito2009EPL} and by Whitney in Ref.~\cite{Whitney2014}, who raised the question of what electronic structure would provide the highest efficiency (or, equivalently, the smallest dissipation) beyond linear response for a given finite power output. Whitney's analysis showed that a boxcar transmission function is needed to achieve maximum efficiency at given average power output.
However, at the microscopic scale, not only the average, but also fluctuations become extremely relevant and need to be properly taken into account in order to determine the performance of thermal machines. The quest towards characterizing their properties during out-of-equilibrium processes and quantifying their connection with dissipation has been a leitmotif throughout the development of stochastic thermodynamics. \cite{Verley2014a,Pietzonka2017,Denzler2019,Denzler2020b,Miller2020,Miller2020b,Denzler2021,Benenti2020}. 
Recently, an important result was found, the Thermodynamic Uncertainty Relations (TURs)~\cite{Barato2015,Gingrich2016,Horowitz2019}. In their simplest form, originally found and proven for Markov processes at a steady-state regime, these inequalities take the form
\begin{equation}\label{TUR}
    \frac{\Delta_\mathcal{J}^2}{\mathcal{J}^2} \geqslant \frac{2}{\sigma},
\end{equation}
with $\Delta_\mathcal{J}^2$ denoting the variance of a generic current, $\mathcal{J}$ its corresponding mean value and $\sigma$ the average entropy production rate.
Eq.~\eqref{TUR} expresses a tradeoff between precision, as quantified by the noise to signal ratio of any thermodynamic quantity, and dissipation.
Beside their fundamental importance, TURs have profound implications for the performance and design of microscopic thermal machines. As shown in Ref. ~\cite{Pietzonka2017} TURs provide in fact a new upper bound,~\fix{different from}{complementing} the one imposed by the second law of thermodynamics, on the maximum achievable efficiency which does not depend only on the average power output anymore but also on the power fluctuations. 
With a little algebra, it is in fact straightforward to show that Eq.~\eqref{TUR} when applied to the power output of an engine,  i.e. $\mathcal{J} = P$, translates into
\begin{equation}\label{TURefficiency}
P \leq \frac{\Delta_P^2}{2T_C}\left(\frac{\eta_C}{\eta}-1\right),
\end{equation}
where $\eta$ is the engine's efficiency, $\eta_C$ is Carnot efficiency and $T_C$ is the temperature of the cold reservoir.  
On the one hand,  Eq.\eqref{TURefficiency}, shows that reaching Carnot's efficiency at finite power is possible but at the cost of diverging power fluctuations, i.e. zero engine's reliability. More generally, this relation highlights how fluctuations quantitatively affect the maximum power achievable at a given efficiency (or, equivalently, at a given amount of dissipation), thus showcasing their importance towards the achievement of optimal nanoscale devices.

Recent studies, however, proved how Eq.~\eqref{TUR} could be violated precisely in thermoelectric devices ~\cite{Agarwalla2018,Ptaszynski2018a,Saryal2019a,Guarnieri2019,Liu2019,Cangemi2020,
Ehrlich2021,Kalaee2021,Saryal2020,Liu2020,Potanina2021,Brandner2018a}, opening up~\fix{to }{}the possibility that quantum mechanics could in principle be exploited to further curb down fluctuations and thus achieve higher performances. 
In this regard, Ref.~\cite{Ehrlich2021} characterized the fluctuations of the power stemming from a boxcar transmission function, i.e. the one yielding the maximum efficiency at finite power output, demonstrating that arbitrarily large violations could be obtained in a thermoelectric device at very high chemical potential gradients. This very naturally leads to the following relevant question, that takes the baton of the research line outlined above: \textit{what is the maximum precision (i.e. the minimum amount of fluctuations) achievable for a given efficiency and output power in a thermoelectric device?}
The current lack of a universal quantum formulation of TURs does not allow in fact to provide a definitive answer to this interrogative; partial answers were provided in the context of thermoelectric systems in linear response regime~\cite{MacIeszczak2018} and in next leading order in the biases ~\cite{Guarnieri2019,Agarwalla2018,Saryal2019a}, as well as for specific sub-classes of fluctuation relations (known as Evan-Searles-Jarzynski-Wojcik fluctuation relations) ~\cite{Merhav2010,Proesmans2019a,Potts2019,VanVu2019,Timpanaro2019,Hasegawa2019,Ray2023,timpanaro-TUR-2024}. 
In this work, we solve this challenging problem by precisely answering the above question. We achieve this by developing a general method to solve concave optimization problems, that we then apply to find the optimal transmission function for which one has \textit{the smallest possible variance $\Delta_{\mathcal{J}}^2$, for fixed $\mathcal{J}$ and $\sigma$}. The answer is remarkably given by a collection of boxcar functions, with boundary positions determined by the values of $I$ and $\sigma$. 
This result can finally be interpreted in the spirit of TURs as the ultimate lower bound to any noise to signal ratio for~\fix{}{non-interacting} thermoelectric devices, valid for arbitrary intensities of the biases. Our work thus provides a\fix{ comprehensive}{n important step forwards in the} solution to the long-standing question of the ultimate bound on the performance of thermoelectric devices, taking into account both the average power and the power fluctuations. Furthermore, our result provides a bound which is always tighter than~\eqref{TUR} close to the linear response regime, reducing to it only in some particular cases.  We finally illustrate our findings in a double quantum dot system coupled between two fermionic reservoirs at different temperatures and chemical potentials.

\section{Landauer-B\"uttiker framework}
We consider a meso- or nanoscale non-interacting quantum system, e.g. a quantum dot array, simultaneously coupled to two macroscopic fermionic reservoirs at different inverse temperatures $\beta_i = 1/T_i$ and chemical potentials $\mu_i$ ($i = L,R$), as depicted in Fig.~\ref{fig:drawing}.
Within the Landauer-B\"uttiker formalism, the (non-equilibrium) steady-state regime is fully characterized by a transmission function $\mathcal{T}(\epsilon) \in [0,1]$.
The particle and energy currents are given by~\cite{Datta1997a}  
\begin{equation}\label{IJ}
    I = \int d\epsilon~\mathcal{T}(\epsilon)~\Delta f(\epsilon), 
    \qquad 
    J = \int d\epsilon~\mathcal{T}(\epsilon)~\epsilon~\Delta f(\epsilon), 
\end{equation}
where $\Delta f(\epsilon) = f_L(\epsilon) - f_R(\epsilon)$,
with $f_i(\epsilon) = (e^{\beta_i(\epsilon- \mu_i)}+1)^{-1}$  denoting the Fermi-Dirac distributions of the left and right reservoirs.

The entropy production rate is given by~\cite{Yamamoto2015a} 
\begin{equation}
    \sigma = -\delta_\beta J + \delta_{\beta \mu} I \geqslant 0,
\end{equation}
where $\delta_\beta = \beta_L - \beta_R$ and $\delta_{\beta\mu} = \beta_L \mu_L - \beta_R \mu_R$. 
Finally, fluctuations around the mean values can be obtained by means of the Levitov-Lesovik full counting statistics formalism~\fix{}{if electron-electron interactions are neglected and the long time limit is considered}~\cite{Levitov1993}.
For concreteness, we focus  on the variance of the particle current $\Delta_I^2$, which is given by 
\begin{IEEEeqnarray}{rCl}\label{var}
    \Delta_I^2 &=&\!\!\! \int d\epsilon~ \mathcal{T}(\epsilon)\left[f_L(\epsilon) + f_R(\epsilon) - 2f_R(\epsilon)f_L(\epsilon) - \mathcal{T}(\epsilon)\Delta f^2(\epsilon) \right],
    \notag\\[0.2cm]
    &=& \int d\epsilon~\Big\{ \mathcal{T}(\epsilon) g(\epsilon) + \mathcal{T}(\epsilon) \big[1 - \mathcal{T}(\epsilon)\big] \Delta f(\epsilon)^2\Big\},
\end{IEEEeqnarray}
where we have introduced for convenience of notation the quantity
\begin{equation}
    g(\epsilon) = f_L(\epsilon)\left[1-f_L(\epsilon)\right] + f_R(\epsilon)\left[1-f_R(\epsilon)\right].
\end{equation}
Similar considerations can be made for the fluctuations of other currents, e.g. of energy.
Moreover, all ideas can be readily extended to systems involving multiple transport channels, such as spin-dependent transmission functions. 

From now on we will focus on the particle current $I$ and our goal will be to find the transmission function which minimizes $\Delta_I^2$ for fixed $I$ and $\sigma$.
Since $\sigma = -\delta_\beta~J + \delta_{\beta \mu}~I$, one can equivalently fix $I$ and $J$.
In fact, our main theorem below  holds for an arbitrary number of constraints, provided they are linear in $\mathcal{T}(\epsilon)$.

\begin{figure}
    \centering
    \includegraphics[width=0.48\textwidth]{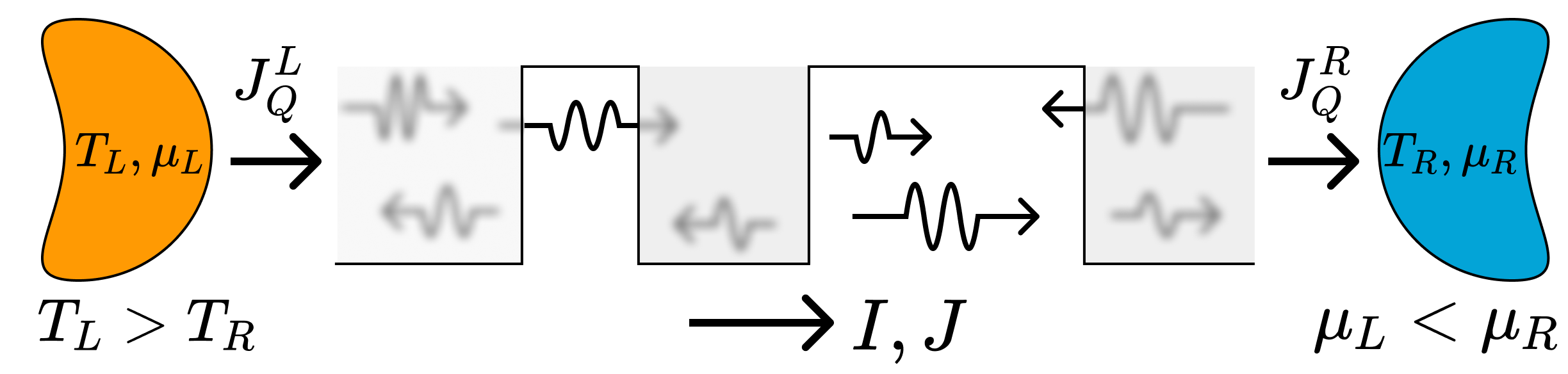}
    \caption{The transport properties across a thermoelectric device is determined by a transmission function $\mathcal{T}(\epsilon)$. In this letter we determine the $\mathcal{T}(\epsilon)$ which minimizes the fluctuations (variance) in the current, for fixed average energy and particle currents.}
    \label{fig:drawing}
\end{figure}

Thermoelectrics can also be viewed as autonomous thermal engines. 
The output power, associated to chemical work, is $P = - \delta_{\mu} I$, where $\delta_\mu = \mu_L - \mu_R$, while the heat current to each bath is given by $J_Q^i = J - \mu_i I$ (Fig.~\ref{fig:drawing}). 
For concreteness, we assume always $T_L > T_R$.
The system will then operate as an engine when $P, J_Q^L,J_Q^R>0$, in which case the efficiency is $\eta = P/J_Q^L$~\cite{Yamamoto2015a}. 
This interpretation allows us to  draw a connection with the seminal results of Ref.~\cite{Whitney2014}, which considered the transmission function maximizing the efficiency for a fixed output power. 
Fixing $P$ is tantamount to fixing $I$. 
Hence, maximizing the efficiency $\eta = - \delta_\mu I/(J - \mu_L I)$ is equivalent to minimizing $J$ for a given power output. 
The problem in~\cite{Whitney2014} can thus be rephrased as which transmission function minimizes $J$ for fixed $I$. 

A crucial difference with respect to our case, however, is that $\Delta_I^2$ is a non-linear (quadratic) functional of $\mathcal{T}$. 
Moreover, it is also a \emph{concave} one. 
Standard tools, such as Lagrange multipliers, therefore do not apply. 
Intuitively speaking, the minimum of a concave function, defined on an interval, is at the boundary of said interval, not somewhere in the middle, as in the convex case. 


A similar argument can be made to the problem at hand, but involving a functional of $\mathcal{T}(\epsilon)$, which can be viewed as a function defined on an infinite dimensional space. 

\section{Main results}

The main result of this Article can be condensed into the following theorem:
\\ \\
{\bf Theorem 1}: \emph{The transmission function $\mathcal{T}(\epsilon)$ which minimizes $\Delta_I^2$ [Eq.~\eqref{var}], for any number of linear constraints, is a collection of boxcars with $\mathcal{T}(\epsilon)$ being either 0 or 1; that is, 
\begin{equation}
    \mathcal{T}_{\rm opt}(\epsilon) = \sum\limits_i \theta(\epsilon-a_i) \theta(b_i - \epsilon),
\end{equation}
where $\theta(x)$ is the Heaviside function and $a_i,b_i \in [-\infty,+\infty]$ are the boxcars boundary points, that are fixed by the linear constraints.
}

While a rigorous proof of this result is provided in the Appendix \ref{supmat:theorem1}, we would like to give a somewhat more heuristic and intuitive explanation. As previously said, the minimum of a concave function defined on an interval is attained at the boundaries. In our case at hand, we need to extrapolate this reasoning to the minimization of a concave functional, i.e. the fluctuations $\Delta^2_I$, which depend on a function $\mathcal{T}(\epsilon)$ instead of a parameter. To do so, we can start by considering a $N$-discretized version of the transmission function as a function from some finite subset of the real line to $[0,1]$. \GGrev{Since for each~\fix{}{of the} $N$~\fix{,}{elements} the transmission function takes values between $0$ and $1$, the region of~\fix{}{possible} values describes a hypercube $[0,1]^N$}, 
where each coordinate represents one of the values of the function, $\mathcal{T}(\epsilon_k)$. In this situation $\Delta^2_I$ is a concave function of the $\mathcal{T}(\epsilon_k)$, whose domain is the hypercube. But this means that the minimum is at the vertices of such hypercube. 
\GGrev{To see this, consider any line segment inside of the hypercube. The function $\mathcal{T}(\epsilon_k)$ restricted to this line segment is concave, univariate and has an interval for a domain, so the minimum in the segment must be one endpoint, which always lies in the boundary of the hypercube (see figure \ref{fig:scheme-concavity}).
Having to simultaneously satisfy this consideration for all segments~\fix{of}{inside} the hypercube\fix{ on which each $\tau(\epsilon_k)$ is defined}{}, one arrives at the conclusion that the minimum is a vertex\fix{ of the hypercube}{}, where the transmission function~\fix{}{only} takes values either $0$ or $1$. 
This thus motivates us to consider boxcar functions as an ansatz.}

\begin{figure}[htbp]
    \centering
    \includegraphics[width=0.3\textwidth]{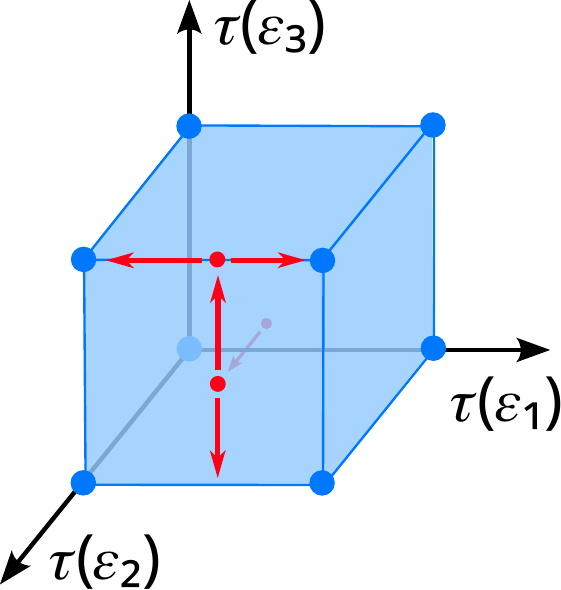}
    \caption{If we discretize the integrals defining $\Delta^2_I$, the fluctuations become a concave function, that always attains its minimum in the vertices, corresponding to boxcar functions.
    }
    \label{fig:scheme-concavity}
\end{figure}


It is worth stressing that Theorem 1 holds true for an arbitrary number of linear constraints. 
The fact that the optimal $\mathcal{T}(\epsilon)$ is a (collection of) boxcar implies that we can substitute $\mathcal{T}^2 = \mathcal{T}$ in Eq.~\eqref{var}, and express it as 
\begin{equation}\label{eq:simple}
\Delta_I^2 = \sum_i \int_{a_i}^{b_i} d \epsilon~g(\epsilon),
\end{equation} 
where $g(\epsilon) = f_L(1-f_L) + f_R(1-f_R)$ and where $\lbrace a_i,b_i\rbrace$ are the boundary points of the boxcars. 
Eq.\eqref{eq:simple} has two very important consequences. The first one is that the original minimization problem, has become a linear one; the second is that the optimization is now turned from finding an optimal continuous function $\mathcal{T}(\epsilon)$ to finding the position of a \textit{finite} set of values, i.e. the boundary points $\lbrace a_i,b_i\rbrace$. The simplification brought by Theorem 1 therefore allows to find those points by means of a Lagrange-Karush-Kuhn-Tucker optimization procedure \cite{lin-prog, vanderberghe}.
To be more concrete, let us focus on the case where $I =I_0$ and $J = J_0$ are fixed. The functional to minimize therefore reads:
\begin{equation}\label{Lagrangian}
    L(\lbrace a_i, b_i\rbrace;\lambda, \eta) = \Delta^2_{I \fix{,\mathrm{opt}}{}} + \lambda(J-J_0) + \eta(I-I_0),
\end{equation}
where $\lambda$ and $\eta$ are the Lagrange multipliers introduced to fix $I$ and $J$.
This leads to our second main result (see Appendix \ref{app:Proof} for a detailed proof):\\ \\
{\bf Theorem 2}: \emph{The position of the optimal boxcars is determined by the regions in energy for which 
\begin{equation}\label{box_equation}
    g(\epsilon) \leqslant (\lambda \epsilon + \eta) \Delta f(\epsilon).
\end{equation}
}
This allows us to identify the boundaries of the boxcars $\lbrace a_i, b_i\rbrace$ when Eq. (\ref{box_equation}) is saturated, i.e. $g(\epsilon) = (\lambda \epsilon + \eta) \Delta f(\epsilon)$. We note that this can also be obtained by minimizing Eq. (\ref{Lagrangian}) through standard Lagrange multipliers. The sign of the inequality in Eq. (\ref{box_equation}), on the other hand, determines which of these endpoints are left/right endpoints of the boxcars. An explanation of its origin is however more subtle and can be found in appendix \ref{supmat:theorem2} where the full Lagrange-Karush-Kuhn-Tucker procedure is done.

Since the problem is now framed in terms of Lagrange multipliers, one should interpret $I(\eta, \lambda)$ and $J(\eta, \lambda)$ as functions of $\eta, \lambda$.
A given choice of $\eta,\lambda$ fixes a\fix{}{n} optimal boxcar, which in turn fixes $I$ and $J$ (the values of the currents for which said boxcar gives the minimum fluctuations).
In practice, of course, what we want is to  work with fixed $I,J$, which thus requires determining the inverse functions $\eta(I,J)$ and $\lambda(I,J)$. 
This has to be done numerically.
We developed a Python library for doing so, which can be downloaded at~\cite{python}.
The calculations are facilitated by the fact that, as we show in appendix \ref{supmat:monotonic},  $I$ is monotonic in $\eta$, and $J$ in $\lambda$.
Once the optimal boxcar~\fix{}{$\mathcal{T}_{\rm opt}$} is determined for a given $I, J$, the minimal variance $\Delta_{I, {\rm opt}}^2$ is computed from Eq.~\eqref{var}, with $\mathcal{T} = \mathcal{T}_{\rm opt}$.
Since the latter is by construction the smallest possible variance out of all transmission functions, it follows that $\Delta_I^2 \geqslant \Delta_{I, {\rm opt}}^2$ for any other model.
This therefore represents a generalized TUR bound. And, in addition, also establishes which process saturates it.

\section{Physical interpretation and significance}

Our framework allows to find the transmission function with the smallest variance, for a fixed $I$ and $J$.
\fix{In practice one is often interested in situations where $I$ and $J$ stem from a concrete physical model.
That is, they are determined from Eq.~\eqref{IJ} by some given transmission function $\mathcal{T}(\epsilon)$. 
One may then ask how does the variance associated to this transmission function fare with respect to the optimal one (obtained from Theorems 1 and 2)?}
{This optimal transmission function will be given by a collection of boxcars, which raises the natural question of how they could be achieved (or approximated) in a real system. While the full answer to this question is outside the scope of the current work, it is known that a chain of quantum dots can be used to achieve an isolated boxcar \cite{whitney-boxcar-construction, Ehrlich2021}. Nevertheless, even if a given optimal transmission function turns out to be physically impossible, the bound derived from it is still valid.

Another question that arises is how the optimal bound compares with the variance obtained from a concrete physical model. That is, if $I, J$ and $\Delta^2_I$ are determined from a given transmission function $\mathcal{T}(\epsilon)$, how do the boxcars obtained from Theorems 1 and 2 fare?}
To illustrate this idea, we consider the problem studied in Ref.~\cite{Ptaszynski2018a}, which  discussed violations of the TUR~\eqref{TUR} in the case of a resonant double quantum-dot system, characterized by the transmission function
\begin{equation}\label{bridge}
    \mathcal{T}_d(\epsilon) = \frac{\Gamma^2 \Omega^2}{|(\epsilon-\omega+i\Gamma/2)^2-\Omega^2 |^2},
\end{equation}
where $\Gamma$ is the system-bath coupling strength, $\omega$ is the excitation frequency of each dot and $\Omega$ the inter-dot hopping constant. 
For simplicity, we fix $T_L = T_R$ and $\mu_R = -\mu_L = \delta_\mu/2$. 

TUR violations can be quantified by analyzing the Fano factor 
$F = \Delta_I^2 /|I|$.
Since $\sigma = \beta \delta_\mu I$ in this case, we see that the TUR~\eqref{TUR} would correspond to $F \geqslant 2/(\beta\delta_\mu)$.
Violations thus occur when $F\beta\delta_\mu< 2$. 
Results for the double quantum dot are shown in Fig.~\ref{fig:qdot}, in red-solid lines, as a function of $\delta_\mu$. 
The violations are generally small, of at most $1.86$ in this case.
This, of course, depends on the choices of parameters, but other results reported in the literature are roughly of the same magnitude~\cite{Saryal2020,Ptaszynski2018a,Agarwalla2018,Liu2019,Saryal2019a,Ehrlich2021}. 
In contrast, the blue-dashed line represents the Fano factor obtained from Eq.~\eqref{box_equation}. This corresponds to the same values of $I, J$ (and $\sigma$) as the red curve, but with the smallest possible $\Delta_{I,{\rm opt}}^2$ allowed over all transmission functions. 
As can be seen, the variance obtained from our generalized TUR bound is not violated by the double quantum dot. Furthermore, $F\beta\delta_\mu$ is now  monotonically decreasing with $\delta_\mu$, and \emph{tends to zero} at infinite bias. 
Consequently, far from linear response, arbitrary violations of the TUR are possible. 
Conversely, close to $\delta_\mu \sim 0$, one recovers $F\beta\delta_\mu = 2$.

\begin{figure}
    \centering
    \includegraphics[width=0.4\textwidth]{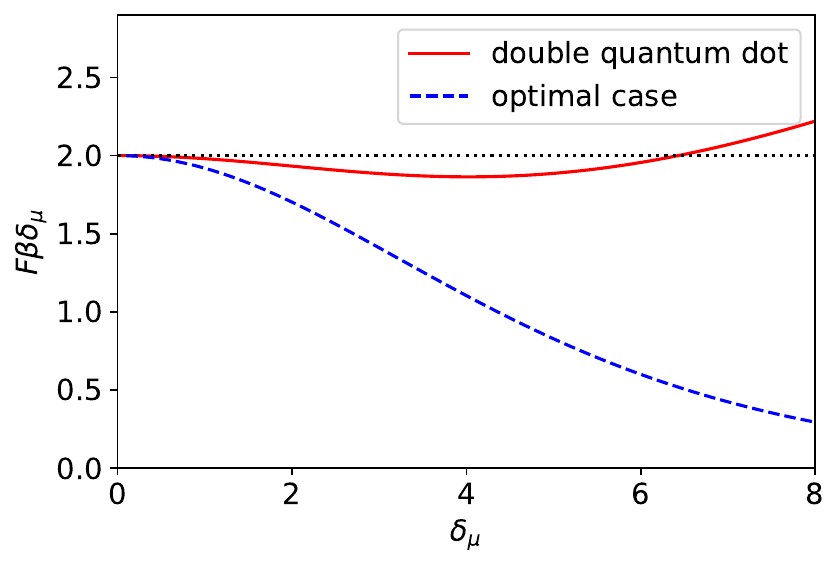}
    \caption{Red-solid:  Fano factor  $F = \Delta_I^2/I$, in units of $(\beta \delta_\mu)^{-1}$, for the double quantum-dot in Eq.~\eqref{bridge}, as a function of $\delta_\mu$. Violations of the TUR~\eqref{TUR} occur when this falls below 2~\cite{Ptaszynski2018a}. 
    Blue-dashed: 
    Fano factor for the minimal variance process, computed using our framework.  
    Parameters: $\Gamma = 0.1$, $\Omega = 0.05$, $\omega = 0$, $\beta_L = \beta_R = 1$ and $\mu_R = -\mu_L = \delta_\mu/2$.
    }
    \label{fig:qdot}
\end{figure}

\section{Linear regime}

The results above illustrate that, far from linear response, it is impossible to bound $\Delta_I^2$ in terms only of $I$ and $\sigma$.
In fact, our results just showed that far from linear response, arbitrary violations of the TUR are possible (similarly to what was found in~\cite{Ehrlich2021}). 
Conversely, one may naturally ask whether the situation simplifies in the linear regime.
We parametrize 
$\beta_L = \beta + \delta_\beta/2$, 
$\beta_R = \beta - \delta_\beta/2$, and  
$\beta_L \mu_L = \beta \mu + \delta_{\beta\mu}/2$, 
$\beta_R \mu_R = \beta \mu - \delta_{\beta\mu}/2$.
Using this in the definition of both $g$ and $\Delta f$ yields, up to the order we are considering
\begin{equation}
    g(\epsilon) \simeq 2 f(\epsilon)(1-f(\epsilon)),
    \label{eq:g-linear-response}
\end{equation}
\begin{equation}
    \Delta f (\epsilon) \simeq \Sigma(\epsilon)f(\epsilon)(1-f(\epsilon)), 
    \label{eq:f-linear-response}
\end{equation}
where
$f(\epsilon) = (e^{\beta(\epsilon-\mu)}+1)^{-1}$ is the Fermi-Dirac distribution associated to the mean temperature and chemical potential and $\Sigma(\epsilon) = \delta_{\beta \mu}-\delta_\beta \epsilon$.
If we consider a general current, then its intensity and fluctuations are given respectively by \cite{Timpanaro2023}:

\begin{equation}
    \mathcal{J}_h = \int d\epsilon~h(\epsilon)~\mathcal{T}(\epsilon)~\Delta f(\epsilon),
    \label{eq:general-intensity}
\end{equation}
\begin{equation}
    \Delta^2_{\mathcal{J}_h} = \int d\epsilon~h^2(\epsilon)~\Big\{ \mathcal{T}(\epsilon) g(\epsilon) + \mathcal{T}(\epsilon) \big[1 - \mathcal{T}(\epsilon)\big] \Delta f(\epsilon)^2\Big\}.
    \label{eq:general-fluctuation}
\end{equation}
So for example, the entropy production is $\sigma = \mathcal{J}_{\Sigma}$. Using theorem 1, we get to the conclusion that for fixed values of $\mathcal{J}_h$ and $\sigma$, the transmission function that minimizes $\Delta^2_{\mathcal{J}_h}$ is a collection of boxcars $\mathcal{B}$. So suppose that the optimal transmission function is the boxcar collection that is 1 in $\mathcal{B}$ and define the functional
\begin{equation}
    \Theta[h] = \int_{\mathcal{B}} d\epsilon~h(\epsilon)~f(\epsilon)(1-f(\epsilon)).
    \label{eq:theta-functional}
\end{equation}
Combining equations (\ref{eq:f-linear-response}) and (\ref{eq:general-intensity}) we get $\mathcal{J}_h \simeq \Theta[h\Sigma]$ and $\Delta^2_{\mathcal{J}_h} \simeq 2\Theta[h^2]$. So

\begin{equation}
    \frac{\sigma \Delta^2_{\mathcal{J}_h}}{\mathcal{J}_h^2} \simeq \frac{2\Theta[h^2]\Theta[\Sigma^2]}{\Theta[h\Sigma]^2} \geq 2,
\end{equation}
where we used the Cauchy-Schwarz inequality. This means that up to the order considered, the optimal process obeys $\nicefrac{\Delta^2_{\mathcal{J}_h}}{\mathcal{J}_h^2} \geq \nicefrac{2}{\sigma}$ which is the classical TUR in equation (\ref{TUR}). Since we also considered a generic current, this implies that we recover this relation in the linear response regime.

To get a more concrete example, we consider again the particle and energy currents: $I = \mathcal{J}_1$ and $J = \mathcal{J}_{\epsilon}$. So substituting and expanding $\Sigma$ leads to

\begin{equation}
    \frac{\sigma \Delta^2_{I}}{I^2} \simeq \frac{2\Theta[1]\Theta[\Sigma^2]}{\Theta[\Sigma]^2} = 2 + \frac{2\delta_\beta^2}{I^2}(\Theta[1] \Theta[\epsilon^2] - \Theta[\epsilon]^2).
    \label{lrt}
\end{equation}

So we can see that if $\delta_\beta = 0$, as in Fig.~\ref{fig:qdot}, then the linear response limit \gtl{for the bound} is exactly $2$. However, if $\delta_\beta \neq 0$, then using Cauchy-Schwarz's inequality one may show that $\Theta[1] \Theta[\epsilon^2] \geqslant \Theta[\epsilon]^2$, so the rhs of~\eqref{lrt} will be strictly larger than 2. Furthermore, since the gradients are small in linear response then $I$ will also be small, so TUR~\eqref{TUR} is generally loose, unless $\delta_\beta = 0$. (see Figure \ref{fig:fano} for an example)

\begin{figure}
    \centering
    \includegraphics[width=0.9\columnwidth]{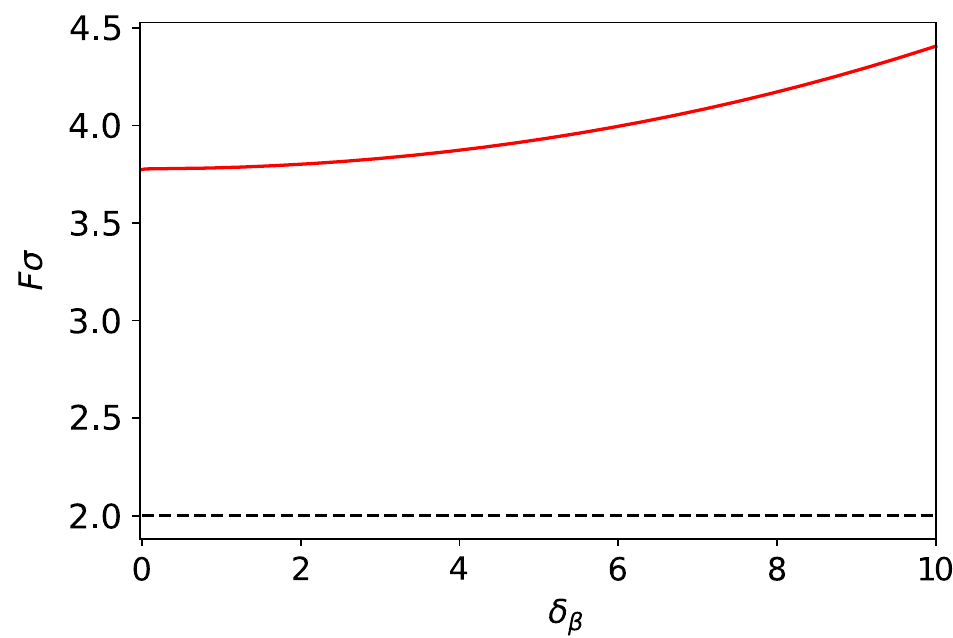}
    \caption{Red-solid: Fano factor $F = \Delta_I^2/I$, in units of $\sigma^{-1}$, for the optimal boxcar given by $\lambda\rightarrow\infty$, $\eta=0$, $\beta_{L(R)} = 10 \pm \nicefrac{\delta_{\beta}}{2}$, $\beta_{L(R)}\mu_{L(R)} = 2\pm \nicefrac{\delta_{\beta\mu}}{2}$, $\nicefrac{\delta_{\beta\mu}}{\delta_{\beta}} = \nicefrac{1}{2}$ as a function of $\delta_{\beta}$.
    Black-dashed: 
    Lower bound given by the TUR~\eqref{TUR}.
    }
    \label{fig:fano}
\end{figure}

\section{Optimal processes}

For each parameter set $(T_L, T_R, \mu_L, \mu_R)$, 
the currents $I$ and $J$ in Eq.~\eqref{IJ} can only take values within a finite interval, irrespective of what the transmission function is. 
An example is shown in Fig.~\ref{fig:boundaries}(a).
Before studying the predictions of Theorems 1 and 2 in more depth, it is thus convenient to establish the boundaries of this region, and then study the corresponding boxcars within them.

The particle current $I$ is bounded by two values, $I_{\rm min}$ and $I_{\rm max}$, which can be found directly from Eq.~\eqref{IJ} by noting that $\Delta f(\epsilon)$ changes sign only once, at the point $\epsilon_0 = \delta_{\beta\mu}/\delta_\beta$. 
Hence, $I_{\rm min/max}$ will have optimal  transmission functions $\mathcal{T}(\epsilon)$ given by boxcars starting at $\epsilon_0$ and extending to either $\pm\infty$. 
These are illustrated in Figs.~\ref{fig:boundaries}(b) and (c).

\begin{figure}[htbp!]
    \centering
    \includegraphics[width=0.5\textwidth]{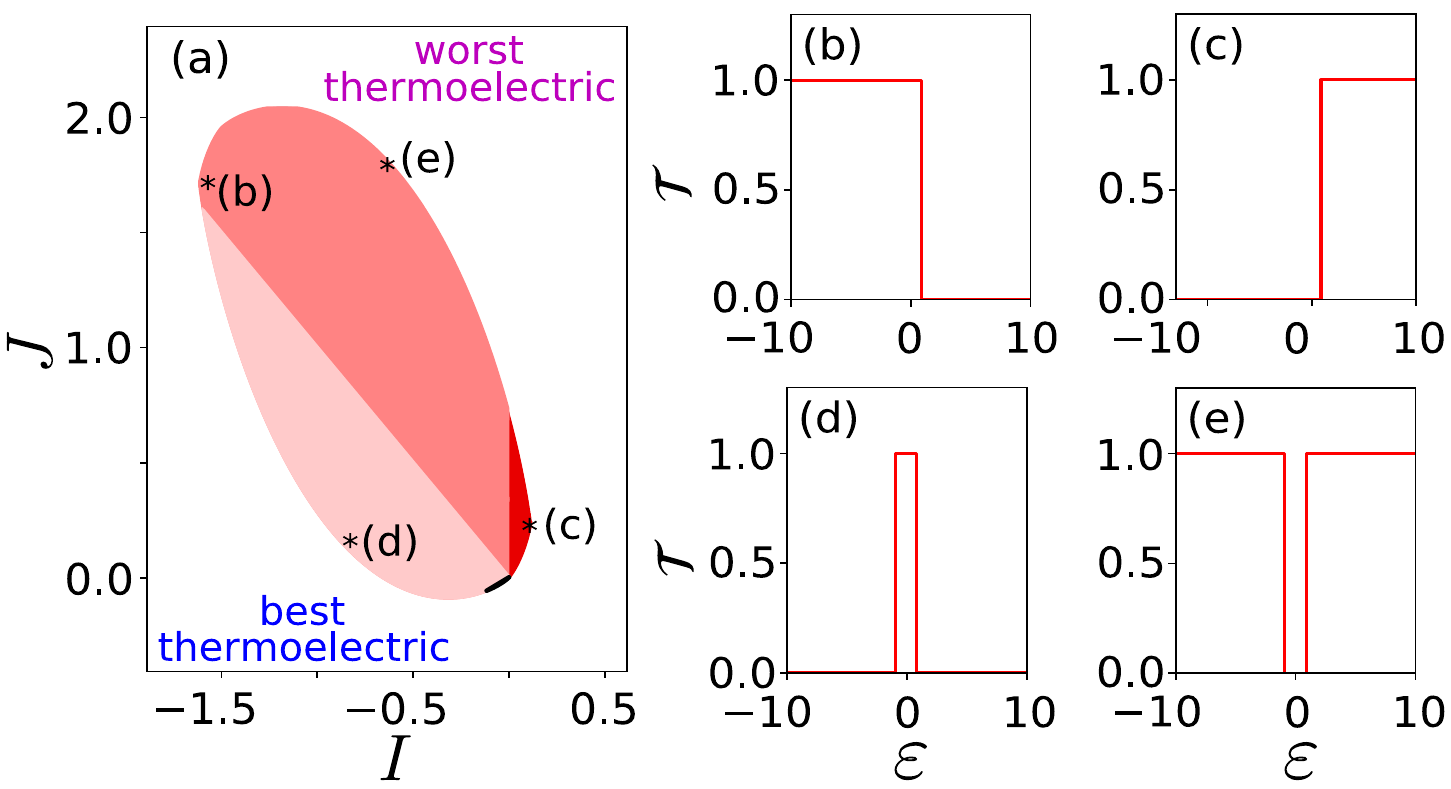}
    \caption{(a) When $(T_L,T_R, \mu_L, \mu_R) = (1,0.2,-1,1/2)$, the allowed values of $I$ and $J$, for any possible transmission function,  must lie within the region formed by the curves $J_{\rm min}(I)$ and $J_{\rm max}(I)$, with $I \in [I_{\rm min}, I_{\rm max}]$~\fix{}{which delimitate the convex region seen in this figure}. 
    \fix{The thick red part of the curve}{The different colors denote four possible regimes for this thermoelectric. The red region to the right} denotes the region where the system operates as an engine ($P = -\delta_\mu I >0$ and $J_Q^L = J- \mu_L I > 0$).~\fix{}{The small black region in the bottom denotes the region where it operates as a refrigerator ($P, J_Q^L < 0$ and $J_Q^R = J - \mu_R I < 0$). In the light pink region to the bottom left it operates as a heater ($P, J_Q^L < 0$ and $J_Q^R > 0$) while in the pink region to the top, it operates as an accelerator ($P< 0$ and $J_Q^L, J_Q^R > 0$).}
    (b)-(e) Examples of boxcars at different points along the boundary. 
    (b) $I_{\rm min}$, 
    (c) $I_{\rm max}$, 
    (d)~\fix{}{point in the } $J_{\rm min}(I)$~\fix{}{curve},
    (e)~\fix{}{point in the } $J_{\rm max}(I)$~\fix{}{curve}.
    }
    \label{fig:boundaries}
\end{figure}

For a given $I \in [I_{\rm min}, I_{\rm max}]$, the energy current $J$ will in turn be bounded by extremal values $J_{\rm min}(I)$ and $J_{\rm max}(I)$. 
At these lines, the solutions of Eq.~\eqref{box_equation}, which minimize $\Delta_I^2$, therefore also extremize $J$.
Since $J$ is monotonic in $\lambda$ \gtl{(see appendix)}, the extrema $J_{\rm min}(I)$ and $J_{\rm max}(I)$ must occur for 
$\lambda \to \pm \infty$. 
\begin{figure}[htbp!]
    \centering
\includegraphics[width=0.45\textwidth]{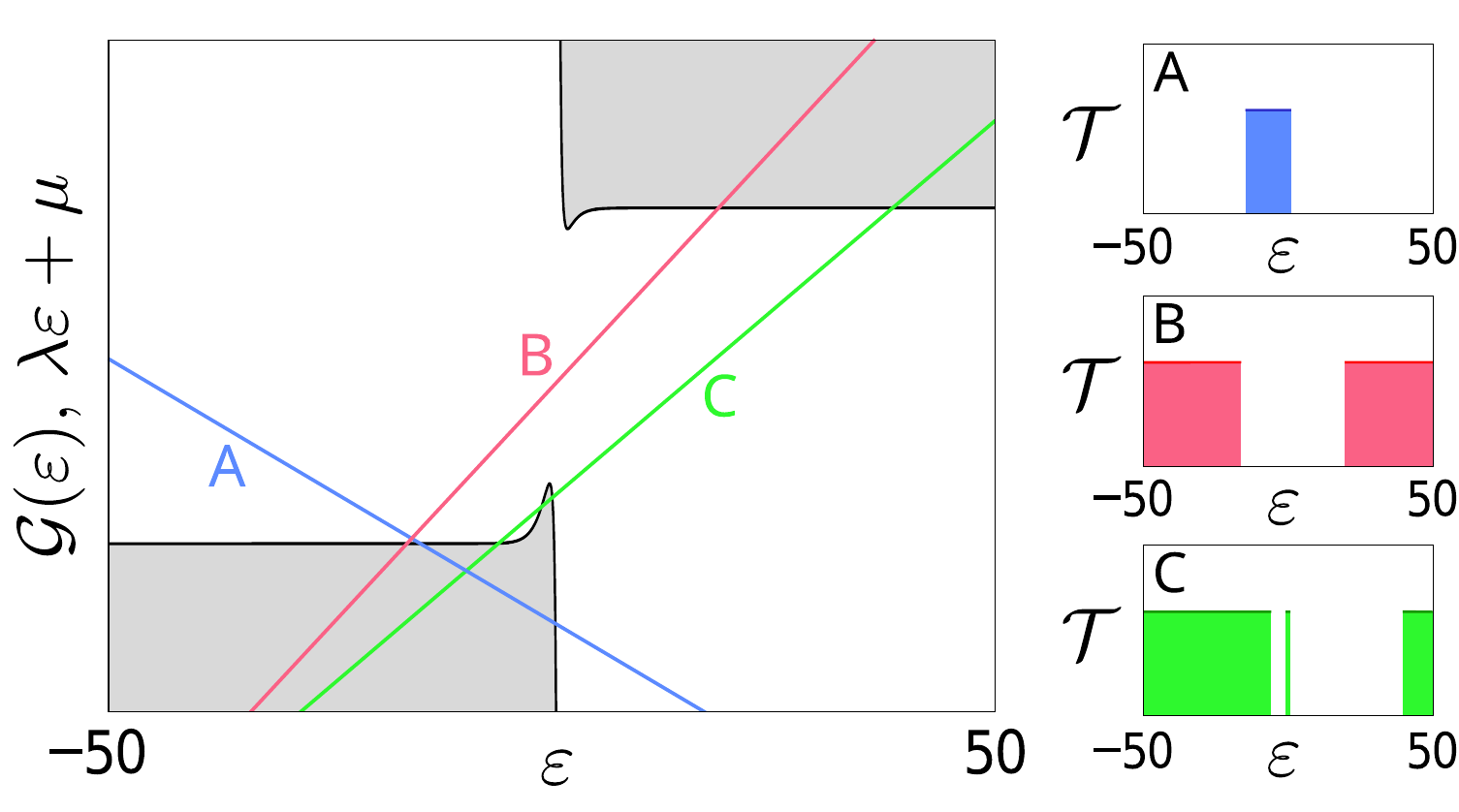}
    \caption{\fix{The crossings between the colored full lines and the graph of $\mathcal{G}$ define the optimal boxcars to the right (the transmission function is 1 in the intersections with the gray region). Each line corresponds to different values of the Lagrange multipliers $(\lambda, \eta)$ and hence to different values of the constraints. The dotted line illustrates one way the boxcars can change qualitatively. In it $\lambda = 0$ and we have a boxcar extending to $+\infty$, however changing $\lambda$ slightly to a positive value makes the boxcar finite, while changing it to a negative value adds a boxcar extending to $-\infty$. The dashed line illustrates the other possibility. In it the line is tangent to the graph of $\mathcal{G}$ and we have two boxcars. One extending to $-\infty$ (to the left out of the graph) and one extending to $+\infty$. In this case, decreasing $\eta$ first creates a third intermediary boxcar, that later merges with the one extending to $+\infty$ once the line becomes tangent to the graph again. The parameters used were $T_L = 1, T_R = 0.2, \mu_L = 0.1, \mu_R = 0.6$.}{Graph of $\mathcal{G}$ and different lines $\lambda \epsilon + \eta$. Each line corresponds to different values of the Lagrange multipliers $(\lambda, \eta)$ and hence to different values of the constraints. The solutions to (\ref{box_equation}) correspond to the intersections of the lines with the gray regions. For each of the colored lines $A$, $B$ and $C$, the corresponding optimal boxcar is displayed on the right. The parameters used were $T_L = 1, T_R = 0.2, \mu_L = 0.1, \mu_R = 0.6$.}}
    \label{fig:gD}
\end{figure}

\begin{figure*}[htbp!]
    \centering
    \includegraphics[height=4.1cm]{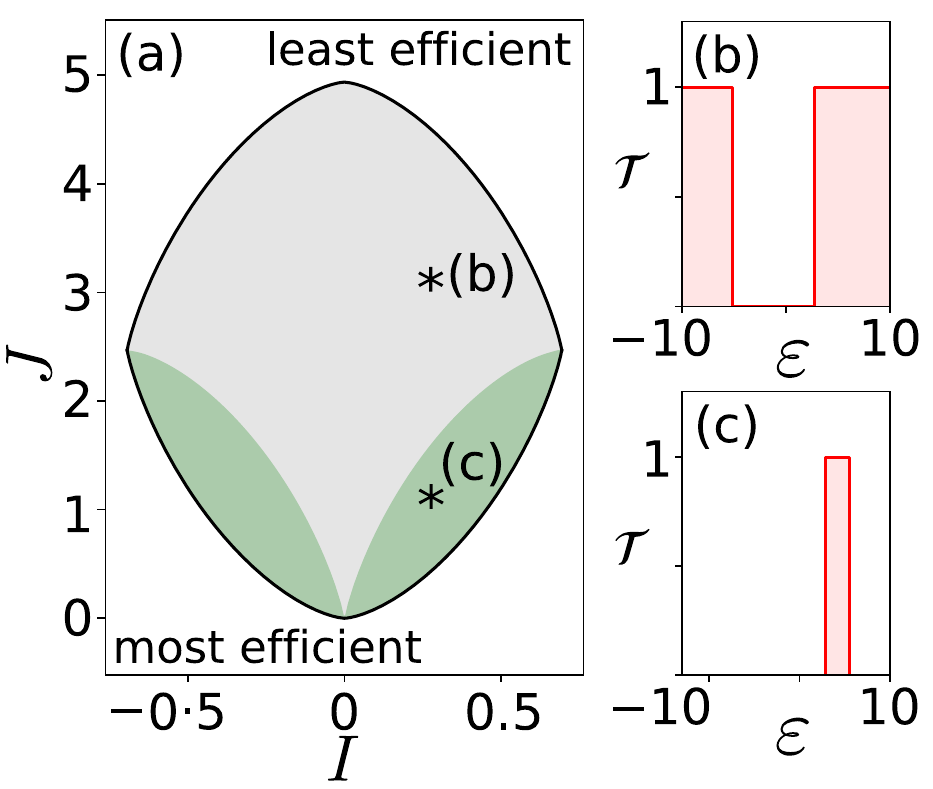}
    \includegraphics[height=4.1cm]{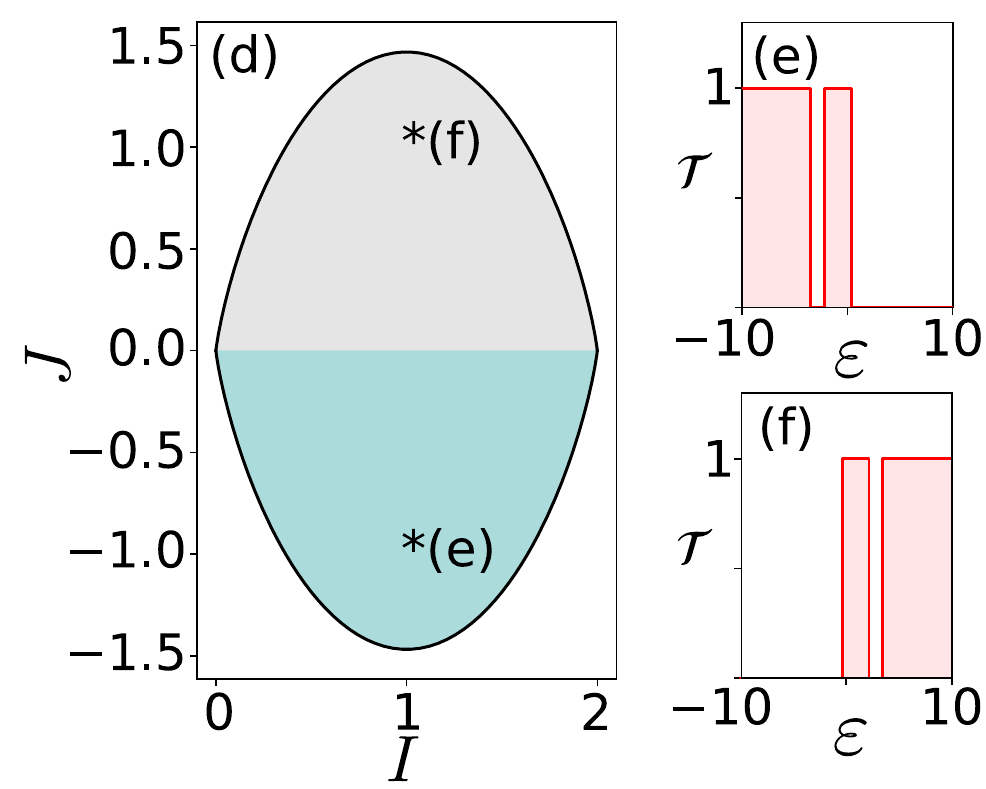}
    \includegraphics[height=4.1cm]{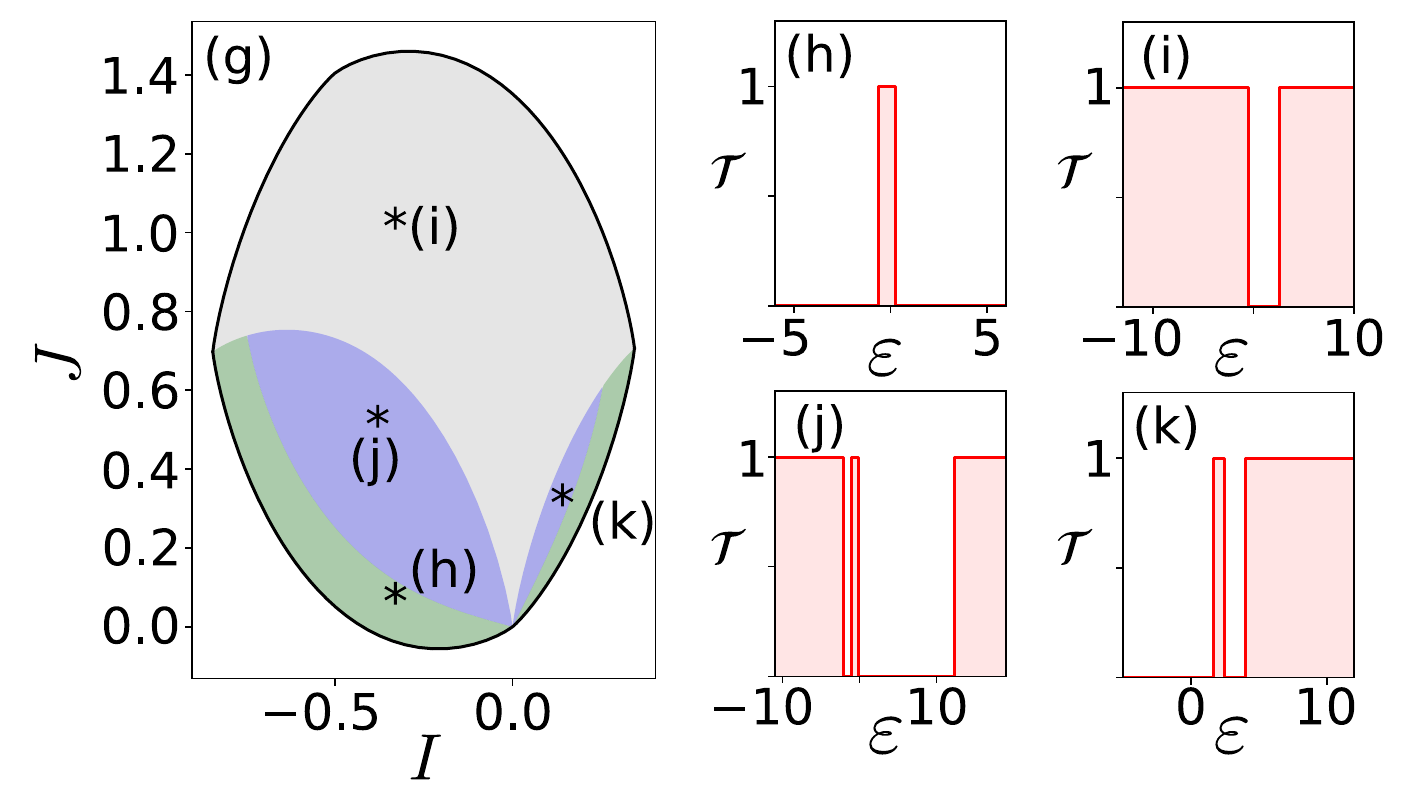}
    \caption{Examples of possible boxcar configurations. 
    The large plots represent the allowed value\fix{}{s} of $I$ and $J$ for different sets $(T_L,T_R, \mu_L, \mu_R)$, with the colors separating regions where the  boxcars are topologically different (see text). 
    The small plots are examples of optimal boxcars in different regions.
    (a)-(c) $(T_L,T_R, \mu_L, \mu_R) = (2,1,0,0)$.
    (d)-(f) $(1, 1, -1,1)$
    (g)-(k) $(1,0.2,0.1,0.6)$.
    }
    \label{fig:example1}
\end{figure*}
Due to the lhs of Eq.~\eqref{box_equation} being always non-negative and finite, in this limit the condition reduces to 
$(\lambda \epsilon+\eta) \Delta f(\epsilon) \geqslant 0$. This expression changes sign twice, at $\epsilon_0$ and $\epsilon_1 = -\eta/\lambda$ (the actual value of $\epsilon_1$ can be determined implicitly as a function of $I$). 
The curve $J_{\rm min}(I)$ corresponds to a compact boxcar in the interval $[\epsilon_0, \epsilon_1]$, as illustrated in Fig.~\ref{fig:boundaries}(d). 
This~\fix{precisely}{curve comprises} the ``best thermoelectric\fix{}{s}'' (or ``most efficient'') in Ref.~\cite{Whitney2014}. 
Conversely, the curve $J_{\rm max}(I)$ (which~\fix{is}{contains} the ``worst thermoelectric\fix{}{s}'') is associated with the complementary boxcar, i.e. one which is 0 in $[\epsilon_0, \epsilon_1]$ and 1 otherwise (Fig.~\ref{fig:boundaries}(e)).
Interestingly, we therefore see that our results also encompass those of~\cite{Whitney2014} as a particular case \gtl{(also meaning that the transmission function that maximizes efficiency for a given particle current also minimizes fluctuations)}. 
However, we call attention to the fact that along these boundaries the system will not necessarily be operating as an engine. 
\fix{It will only do so in the region highlighted by a red thick line in}{The different possible working regimes are identified in} Fig.~\ref{fig:boundaries}(a).

For values of $I,J$  inside the region defined by $J_{\rm min/max}(I)$, the shape of the boxcar must be determined numerically from Eq.~\eqref{box_equation}. This is done by considering for a pair $(\lambda, \eta)$ the functions $\lambda \epsilon + \eta$ and $\mathcal{G}(\epsilon) = \nicefrac{g(\epsilon)}{\Delta f(\epsilon)}$. The \GGrev{intersection} between these functions \GGrev{provide the positions of the boundaries} of the optimal boxcar for given $(\lambda, \eta)$ (see figure \ref{fig:gD} for some examples). \GGrev{This insight also sheds  light on the intricate dependence between the parameters $(T_L, T_R, \mu_L, \mu_R)$, the constraints $(I, J)$ and the possible boxcar topology (e.g. open boundary or closed boxes, single or multiple boxes etc). This in fact depends on the different ways that the line $\lambda \epsilon + \eta$ and the graph of $\mathcal{G}(\epsilon)$ may intersect. In figure \ref{fig:gD} we illustrate this by connecting the color of the curve with the respective resulting transmission function.}

\gtl{Finally, calculating the currents for the parameters $(\lambda, \eta)$ where these qualitative changes can occur leads us to boundaries in the $I-J$ plane of regions inside of which the boxcars have a given qualitative shape.
Several examples are shown in Fig.~\ref{fig:example1}:
(a)-(c) corresponds to $\mu_L = \mu_R$; 
(d)-(f) to $T_L = T_R$;
and (g)-(k) to a bias in both affinities.
In these graphs the regions corresponding to each possible boxcar topology are represented by different colors in  Figs.~\ref{fig:example1}(a),(d),(g). 
}

The largest number of boxcars we have been able to find is 3, as in Fig.~\ref{fig:example1}(j).
This is related to the type of constraints we are imposing. 
If the variance of some other current is studied, or if additional linear constraints are imposed (e.g. in the case of spin-dependent transport channels), a more complex boxcar topology can in principle occur.

As a more concrete example, we can once again consider the double quantum dots and ask ourselves what transmission functions would yield the same currents with optimal variance and in the presence of the same gradients. Since we had $\beta_L = \beta_R = \beta$ and $\mu_R = -\mu_L = \delta_\mu/2$, this scenario is located exactly at the boundary between the gray and cyan regions in Fig.~\ref{fig:example1}(d). 
The optimal transmission function is found to be a single boxcar of the form $[-a/2,a/2]$ 
\footnote{
The width $a$ can be determined by imposing that the current $I$ obtained from $\mathcal{T}_d$ must be the same as that obtained using the boxcar. 
Hence, $a$ must be a solution of 
$\int_{-a/2}^{a/2} d\epsilon~\Delta f(\epsilon) = \int_{-\infty}^\infty d\epsilon ~\mathcal{T}_d(\epsilon)\Delta f(\epsilon)$, 
which generally depends on both $\beta$ and $\delta_\mu$.
}.
For boxcars of this form, it is actually possible to determine the optimal Fano factor $F_{\rm opt} = \Delta_{I, {\rm opt}}^2/|I|$ analytically, by directly computing the integrals in Eqs.~\eqref{IJ} and~\eqref{var}.
The result is 
$F_{\rm opt} =  2(1-f_L-f_R)/[\ln f_R(1-f_R)f_L(1-f_L)]$,
where $f_{L(R)} = f_{L(R)}(a/2)$ are the Fermi-Dirac functions of each bath, evaluated at the left (right)-end of the boxcar.

\section{Discussion}

In this paper we have provided a definitive answer to the question of TUR violations in~\fix{}{non-interacting} quantum thermoelectrics. 
We showed that, beyond linear response, no bound exists which relate $\Delta_I^2$ only to $I$ and $\sigma$; instead, the trade-off relation involving these quantities becomes dependent on the system parameters.
Our approach addresses this issue by determining what is the optimal process; 
i.e., out of all possible processes allowed in Nature, which one yields the smallest possible variance for a fixed $I$ and $\sigma$? 
We believe this represents a very insightful question. 
First, it yields more general bounds, not necessarily related only to $I$ and $\sigma$. 
And second, and most importantly, it actually tells us which process is the optimal one.

Recently, there has been growing interest in this kind of approach.
Irrespective of whether or not achieving the optimal process is easy, knowing what it is provides a benchmark which must be satisfied by any other process. 
For instance, Ref.~\cite{Cavina2016} determined the probability distribution maximizing work extraction in the single shot-scenario. 
This was later studied experimentally in~\cite{Maillet2019}, which performed a process that was not exactly optimal, but got extremely close. 

The design of transmission functions in quantum thermoelectrics is a relevant technological problem. 
Our results introduce fluctuations as a new ingredient to the mix. 
It provides, to our knowledge, the only known route for designing transmission functions with target power and efficiency, but also minimizing fluctuations. 

\emph{\textbf{Acknowledgements--}} 
GTL acknowledges the financial support of the  S\~ao Paulo Funding Agency FAPESP (Grants No. 2017/50304-7, 2017/07973-5 and 2018/12813-0), the Gridanian Research Council (GRC), and the Brazilian funding agency CNPq
(Grant No. INCT-IQ 246569/2014-0). G. G. acknowledges support from FQXi and DFG FOR2724 and also from the European Union Horizon 2020 research and innovation programme under the Marie Sklodowska-Curie grant agreement No. 101026667.


\appendix

\widetext




\section{Proof of Theorems 1 and 2}
\label{app:Proof}
 
The problem consists in minimizing the variance $\Delta_I^2$ [Eq.~(3)] subject to a set of linear (in $\mathcal{T}$) constraints; in our case, fixed $I$ and $J$. 
Formally, this may be written as the following mathematical optimization problem:
\begin{framed}
\[
\mbox{Minimize } \int_{-\infty}^{\infty} f(x) \beta(x) - f(x)^2 \gamma(x)^2 \ud x, \quad\mbox{subject to}
\]
\begin{equation}
\int_{-\infty}^{\infty} f(x) \alpha_i(x) \ud x = \phi_i,\,\,\, i = 1,\ldots, n\quad\mbox{and}\quad 0 \leq f(x) \leq 1\,\,\forall\,\, x\in\mathds{R}.
\label{eq:opt-def}
\end{equation}
\end{framed}
This is a concave functional of $f(x)$, which therefore requires specific methods in order to be tackled. 

In what follows, we start by fixing the necessary notations and definitions employed throughout this Supplementary Material.
The proofs of Theorems 1 and 2 are then given in Secs.~\ref{supmat:theorem1} and~\ref{supmat:theorem2}, \GGrev{respectively}. 
\GGrev{A series of additional} remarks and details are finally provided in Appendix Sec.~\ref{app:misc}. 

\subsection{Notations, Definitions and Standing Hypotheses}

To make things precise, we make the following hypothesis and definitions:
\begin{itemize}
\item We will denote by $\widetilde{\mathcal{R}}$ the set of functions $f:\mathds{R}\rightarrow \mathds{R}$ that are bounded and such that 
for every compact interval $[a,b]$, $f$ has only finitely many discontinuities in $[a,b]$. Furthermore, we will denote by $\mathcal{R}$ the subset of $\widetilde{\mathcal{R}}$ where~\fix{$\mathrm{Im}(f)\subseteq [0,1]$.}{$0\leq f(x) \leq 1\,\forall\,x$.}
\item $\alpha_i$, $\beta$ and $\gamma$ will always denote functions in $\widetilde{\mathcal{R}}$, such that
\[
\int_{-\infty}^{\infty} |\alpha_i(x)| \ud x, \quad\int_{-\infty}^{\infty} |\beta(x)| \ud x\quad\mbox{and}\quad \int_{-\infty}^{\infty} \gamma(x)^2 \ud x
\]
are all finite and such that $\alpha_i(x) \neq 0$, $\beta(x) \neq 0$ and $\gamma(x) \neq 0$ almost everywhere.
\item We define the following functionals acting on functions in $\mathcal{R}$:
\[
\mathcal{C}_i[f] = \int_{-\infty}^{\infty} f(x) \alpha_i(x) \ud x,\quad i = 1,\ldots, n,\quad\quad\quad\quad\quad\quad \mathcal{Q}[f] = \int_{-\infty}^{\infty} f(x) \beta(x) - f(x)^2 \gamma(x)^2 \ud x,
\]\[
\mathcal{L}[f] = \int_{-\infty}^{\infty} f(x) (\beta(x) - \gamma(x)^2) \ud x, \quad\quad\quad\mbox{and}\quad\quad\quad \mathcal{B}[f] = \int_{-\infty}^{\infty} \gamma(x)^2 f(x)(1-f(x)) \ud x.
\]
The functionals $\mathcal{C}_i$ are the ones that give us the constraints in (\ref{eq:opt-def}), while $\mathcal{Q}$ is the one to be optimized. As we will see further on, the functional $\mathcal{L}$ can be regarded as a linearized version of $\mathcal{Q}$. Finally, $\mathcal{B}[f]$ can be used as a measure of how far $f$ is from being a boxcar.
\item We will be using the following jargon from optimization theory \cite{lin-prog, vanderberghe}:
\begin{itemize}
\item A feasible point (function) of an optimization problem is a point (function) that obeys all the constraints.
\item The feasible region is the set of all feasible points (functions).
\item The optimal value is the value of the desired extremum.
\item An optimal point (function) is a feasible point (function) that attains the desired extremum.
\fix{}{\item In an optimization problem, the function/quantity we want to either minimize or maximize is called the objective function.}
\end{itemize}
\item The feasible region for the problem (\ref{eq:opt-def}) will be denoted $\mathcal{F}$:
\[
\mathcal{F} = \left\{f\in\mathcal{R}\mbox{ such that }\mathcal{C}_i[f] = \phi_i\,\forall\,i\right\}.
\]
\item We will denote by $\mathsf{Q}$ and $\mathsf{L}$ the optimal values  minimizing $\mathcal{Q}$ and $\mathcal{L}$ respectively, given the constraints $\mathcal{C}_i$:
\[
\mathsf{Q} = \inf_{f\in\mathcal{F}} \mathcal{Q}[f] \quad\quad\mbox{and}\quad\quad 
\mathsf{L} = \inf_{f\in\mathcal{F}} \mathcal{L}[f].
\]
Note that the hypothesis made about $\beta$ and $\gamma$ imply that $\mathsf{Q}$ and $\mathsf{L}$ must be finite when $\mathcal{F} \neq \varnothing$.
\item We will denote by $\mathcal{D}(F)$ the set of discontinuities of a function $F:\mathds{R}\rightarrow\mathds{R}$.
\item Finally, we recall the definition of oscillation of a function $f$ in an interval (as used in Analysis):
\[
\omega_f(I) = \sup_{x\in I} f(x) - \inf_{x\in I} f(x).
\]

\end{itemize}


\subsection{\label{supmat:theorem1}Proof of theorem 1}

We start by proving theorem 1\fix{, n}{. N}amely\fix{}{:}

\begin{theorem}
The transmission function $\mathcal{T}(\epsilon)$ which minimizes $\Delta_I^2$ [Eq.~\eqref{var}] for any number of linear constraints is a collection of boxcars, with $\mathcal{T}(\epsilon)$ being either 0 or 1; that is, 
\[
\mathcal{T}_{\rm opt}(\epsilon) = \sum\limits_i \theta(\epsilon-a_i) \theta(b_i - \epsilon),
\]
where $\theta(x)$ is the Heaviside function and $a_i, b_i \in [-\infty,+\infty]$ are the boxcar\fix{s}{} boundary points that are fixed by the linear constraints in question.
\label{teor:boxcar}
\end{theorem}

This will actually be accomplished by proving the following statement (using the notations we just defined):

\begin{equation}
\text{If $f\in\mathcal{F}$ is such that $\mathcal{Q}[f] \leq \mathsf{Q} + \delta$, then $\mathcal{B}[f] \leq \delta$}.
\end{equation}

\begin{proof}
As a first step we will show that if $f \in \mathcal{F}$, then for all $\varepsilon > 0$ there exist $g, h \in \mathcal{F}$ such that $\mathcal{Q}[g] \leq \mathcal{Q}[f]$, $\mathcal{L}[h] \leq \mathcal{L}[f]$ and $\mathcal{B}[g], \mathcal{B}[h] \leq \varepsilon$. To see this, let us consider $\eta > 0$, an arbitrary positive number and $[a,b]$ a compact interval. Since $f\in\mathcal{R} $, then $\mathcal{D}(f) \cap [a,b]$ is finite and there exist $D$~\fix{}{(with $D\,= |\mathcal{D}(f) \cap [a,b]|)$} open intervals $I_1, \ldots, I_D$, all of which have measures less than $\eta$ and such that if $\mathcal{I} = \bigcup_i I_i$, then $\mathcal{D}(f) \cap [a,b] \subseteq \mathcal{I}$. As a consequence $f$ has no discontinuities on $\mathcal{J} = [a,b]\setminus\mathcal{I}$. As such, every point $x\in\mathcal{J}$ has an open interval~\fix{$N_x$}{$U_x$} containing it, such that the oscillation of $f$ in this interval,~\fix{$\omega_f(N_x)$}{$\omega_f(U_x)$} is less than $\eta$. Since the~\fix{$N_x$}{$U_x$} form an open cover of $\mathcal{J}$ and $\mathcal{J}$ is compact, then there exists a subcover that is finite: $M_1, \ldots, M_C$.

Let us consider a partition $\Pi$ of the interval $[a,b]$, such that all the endpoints of the $I_k$ and the $M_k$ intervals that lie in $[a,b]$ are included and none of the subintervals has a measure larger than $\eta$~\fix{}{(see figure \ref{fig:partition1} for an illustration of such partition)}. We will denote by $J_i$ the $N$ open subintervals of $\Pi$ such that $\omega_f < \eta$ and by $K_i$ the remaining ones. Note that the measure of the union of the $J_i$ must be at least $b-a-D\eta$ and hence the measure for the union of the $K_i$ is at most $D\eta$.

\clearpage

\begin{figure}[htbp!]
 \centering
    \includegraphics[width=0.6\textwidth]{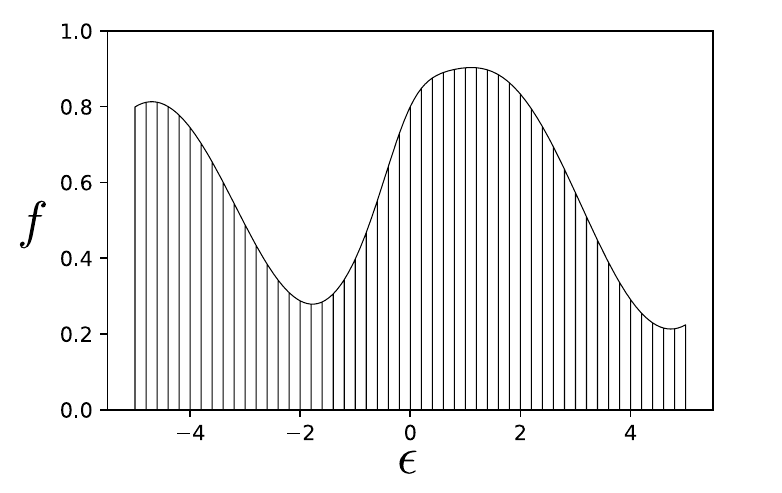}
    \caption{\fix{}{A possible partition for a function $f$ in $\mathcal{R}$. The interval $[a,b]$ being considered is $[-5,5]$ (we are ignoring $f$ outside of this interval).}}
    \label{fig:partition1}
\end{figure}

We then define the constants
\[
f^{(-)}_i = \inf_{x\in J_i} f(x),\quad\quad f^{(+)}_i = \sup_{x\in J_i} f(x), \quad\quad\mbox{and}\quad\quad m_i = 1 - f^{(+)}_i + f^{(-)}_i,
\]
the integrals

\[
a_{i,j} = \int_{J_j} \alpha_i(x) \ud x,\quad\quad\quad\quad\quad\quad
b_j = - \int_{J_j} f^{(-)}_j (\beta(x) + (f^{(-)}_j - 2f(x))\gamma(x)^2) \ud x,
\]\[
c_j = m_j \int_{J_j} (\beta(x) + 2(f^{(-)}_j - f(x))\gamma(x)^2) \ud x,\quad\quad\quad\quad\quad\quad
d_j = m_j^2 \int_{J_j} \gamma(x)^2 \ud x,
\]\[
b'_j = - \int_{J_j} f^{(-)}_j (\beta(x) - \gamma(x)^2) \ud x,\quad\quad\quad\quad\quad\quad c'_j = m_j \int_{J_j} (\beta(x)-\gamma(x)^2) \ud x,
\]
and the functions

\[
F_{\lambda}(x) = \left\{
\begin{array}{ll}
f(x) - f^{(-)}_i + m_i \lambda_i & \mbox{if }x\in J_i \\
f(x) & \mbox{otherwise}
\end{array}.
\right.
\]
It follows that if $0 \leq \lambda_i \leq 1$\fix{}{$\forall\,i$}, then $F_{\lambda}\in\mathcal{R}$ and we also have
\[
\mathcal{C}_i[F_{\lambda}] = \phi_i + \sum_j a_{i,j} (m_j \lambda_j - f^{(-)}_j),\quad\quad\quad\quad\quad\quad
\mathcal{Q}[F_{\lambda}] = \mathcal{Q}[f] + \sum_j (b_j + c_j \lambda_j - d_j \lambda_j^2)\quad\quad\mbox{and}
\]\[
\mathcal{L}[F_{\lambda}] = \mathcal{L}[f] + \sum_j (b'_j + c'_j \lambda_j).
\]
Furthermore, if $\mu_i = \nicefrac{f^{(-)}_i}{m_i}$ for every $i$ such that $m_i\neq 0$ ($\mu_i$ can be any value in $[0,1]$ if $m_i = 0$), then we have $F_{\mu} = f$. Consider then the following concave programming problems:

\begin{framed}
Minimize
\[
\sum_j (b_j + c_j \lambda_j - d_j \lambda_j^2) \mbox{ subject to the constraints}
\]
\begin{equation}
\sum_j a_{i,j} (m_j \lambda_j - f^{(-)}_j) = 0\quad\quad\mbox{and}\quad\quad 0\leq\lambda_i\leq 1\,\forall\, i.
\label{concave-programm-2}
\end{equation}
\end{framed}

and

\begin{framed}
Minimize
\[
\sum_j (b'_j + c'_j \lambda_j) \mbox{ subject to the constraints}
\]
\begin{equation}
\sum_j a_{i,j} (m_j \lambda_j - f^{(-)}_j) = 0\quad\quad\mbox{and}\quad\quad 0\leq\lambda_i\leq 1\,\forall\, i.
\label{linear-programm}
\end{equation}
\end{framed}

Since $F_{\mu} = f$, then $(\mu_1, \ldots, \mu_N)$ is a feasible point for both (\ref{concave-programm-2}) and (\ref{linear-programm}) corresponding in both cases to a value 0 for the objective function. As a consequence, if $\nu = (\nu_1, \ldots ,\nu_N)$ and $\sigma = (\sigma_1, \ldots ,\sigma_N)$ are optimal points for (\ref{concave-programm-2}) and (\ref{linear-programm}) respectively, it follows that $\mathcal{C}_i[F_{\nu}] = \mathcal{C}_i[F_{\sigma}] = \phi_i$, while $\mathcal{Q}[F_{\nu}]\leq \mathcal{Q}[f]$ and $\mathcal{L}[F_{\sigma}]\leq \mathcal{L}[f]$. So we need to show that we can choose $\eta$ and $[a,b]$ in such a way that $\mathcal{B}[F_{\nu}], \mathcal{B}[F_{\sigma}] \leq \varepsilon$. In what follows $\xi$ means either $\nu$ or $\sigma$ (the argument is identical). To see why we can make $\mathcal{B}[F_{\xi}] \leq \varepsilon$, note that in both cases the feasible region $R$ will be the intersection of a hypercube $[0,1]^N$ and $n$ hyperplanes defining the linear constraints. Since there is always an extremal point of the feasible region where the optimum of a concave problem is attained \cite{tuy-horst}~\fix{}{(an intuitive argument for this was provided in the main text)}, then one needs only to look for the extremal points of $R$. Since $R$ is a polytope, the extremal points are the vertexes, which in this case are all points where at least $N-n$ coordinates are either 0 or 1, implying that $\xi$ has at least $N-n$ coordinates that are either 0 or 1~\fix{}{(see figure \ref{fig:polytope} for an illustration).}

\begin{figure}[htbp!]
 \centering
    \includegraphics[width=0.4\textwidth]{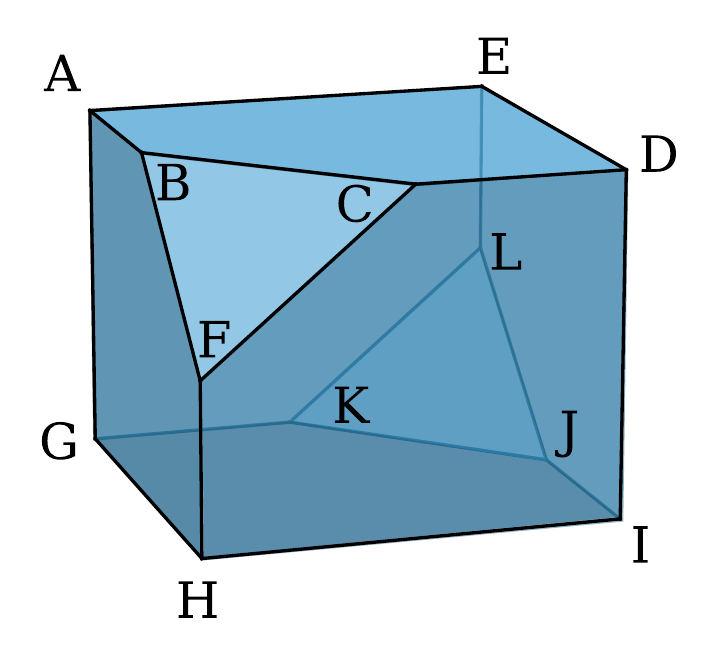}
    \caption{\fix{}{The intersection of the hypercube $[0,1]^4$ and the hyperplane $x + 2y - z + w = \nicefrac{3}{2}$ embedded in 3 dimensions. The vertex list of this polytope is: $A=(0,\nicefrac{3}{4},0,0)$, $B=(0,1,\nicefrac{1}{2},0)$, $C=(\nicefrac{1}{2},1,1,0)$, $D=(1,\nicefrac{3}{4},1,0)$, $E=(1,\nicefrac{1}{4},0,0)$, $F=(0,1,1,\nicefrac{1}{2})$, $G=(0,\nicefrac{1}{4},0,1)$, $H=(0,\nicefrac{3}{4},1,1)$, $I=(1,\nicefrac{1}{4},1,1)$, $J=(1,0,\nicefrac{1}{2},1)$, $K=(\nicefrac{1}{2},0,0,1)$, $L=(1,0,0,\nicefrac{1}{2})$. An important point is that in all vertices at most one coordinate is different from 0 or 1, corresponding to the single hyperplane used.}}
    \label{fig:polytope}
\end{figure}

In turn, this implies that
\begin{itemize}
\item If $x\in J_i$ and $\xi_i = 0$, then $0 \leq F_{\xi}\fix{}{(x)} \leq \eta$;
\item If $x\in J_i$ and $\xi_i = 1$, then $1- \eta \leq F_{\xi}\fix{}{(x)} \leq 1$;
\item If $x\in J_i$ but $\xi_i \neq 0, 1$, then we can only say $0 \leq F_{\xi}\fix{}{(x)} \leq 1$.
\end{itemize}

Assuming $\eta<\nicefrac{1}{2}$, it follows that
\[
\mathcal{B}[F_{\xi}] = \int_{-\infty}^{\infty} \gamma(x)^2 F_{\xi}(x)(1-F_{\xi}(x)) \ud x = \underbrace{\int_{-\infty}^{a} \ldots + \int_b^{\infty} \ldots}_{\mathcal{B}_1} + \underbrace{\sum_i \int_{K_i} \ldots}_{\mathcal{B}_2} + \underbrace{\sum_j \int_{J_j} \ldots}_{\mathcal{B}_3},
\]
where we simply broke the integral over different intervals. Since the integral $\int_{-\infty}^{\infty} \gamma(x)^2 \ud x$ is finite (from our hypotheses) and $0 \leq F_{\xi}(x)(1-F_{\xi}(x)) \leq 1$, then it is trivial that we can choose $[a,b]$ such that $\mathcal{B}_1 \leq \nicefrac{\varepsilon}{3}$ by choosing a sufficiently large interval. Since the measure of $\bigcup_i K_i$ is at most $D\eta$ and $\gamma(x)^2$ is bounded (again from our hypotheses), then one can easily bound $\mathcal{B}_2$:

\[
\mathcal{B}_2 = \sum_i \int_{K_i} \gamma(x)^2 F_{\xi}(x)(1-F_{\xi}(x)) \ud x \leq \sum_i \int_{K_i} \Gamma \ud x \leq D\eta\Gamma,
\]
where $\Gamma$ is any upper bound for $\gamma(x)^2$. So we can take $\mathcal{B}_2 \leq \nicefrac{\varepsilon}{3}$ by choosing a sufficiently small $\eta$ (namely $\eta \leq \nicefrac{\varepsilon}{3D\Gamma}$). Finally, $\mathcal{B}_3$ can be bounded in a similar way, because if $\eta < \nicefrac{1}{2}$ and $x\in J_i$ (where $\xi_i = 0$ or 1), then $F_{\xi}(x)(1-F_{\xi}(x)) \leq \eta(1-\eta) \leq \eta$. Furthermore, the measure of the $J_i$ intervals where $\xi_i \neq 0, 1$ is at most $n\eta$. So:

\[
\mathcal{B}_3 = \sum_j \int_{J_j} \gamma(x)^2 F_{\xi}(x)(1-F_{\xi}(x)) \ud x = \!\!\sum_{\substack{j \\ \xi_j = 0\text{ or }1}} \int_{J_j} \gamma(x)^2 F_{\xi}(x)(1-F_{\xi}(x)) \ud x + \!\!\sum_{\substack{j \\ \xi_j \neq 0, 1} }\int_{J_j} \gamma(x)^2 F_{\xi}(x)(1-F_{\xi}(x)) \ud x \leq
\]\[
\leq \sum_{\substack{j \\ \xi_j = 0\text{ or }1}} \int_{J_j} \Gamma \eta \ud x + \!\!\sum_{\substack{j \\ \xi_j \neq 0, 1} }\int_{J_j} \Gamma\ud x \leq \Gamma\eta\left(b -a + n\right).
\]
So we can take $\mathcal{B}_3\leq \nicefrac{\varepsilon}{3}$ by making $\eta$ sufficiently small ($\eta \leq \nicefrac{\varepsilon}{3\Gamma(b-a+n)}$). It follows that by taking $[a,b]$ large enough and $\eta \leq \min\left\{\nicefrac{1}{2}, \nicefrac{\varepsilon}{3D\Gamma}, \nicefrac{\varepsilon}{3\Gamma(b-a+n)}\right\}$ we get $\mathcal{B}[F_{\xi}] \leq \varepsilon$ (so $g = F_{\nu}$ and $h = F_{\sigma}$)\fix{}{. The interpretation of this is that the solutions of the optimization problems (\ref{concave-programm-2}) and (\ref{linear-programm}) lead to functions $F_{\lambda}$ that obey the same constraints as $f$, while having a better objective (in the interpretation of the main paper, $F_{\lambda}$ is a transmission function with the same currents but a smaller variance) and being closer to a boxcar function (because of the structure of the solution). This is illustrated in figure \ref{fig:partition2}.}

\begin{figure}[htbp!]
 \centering
    \includegraphics[width=0.48\textwidth]{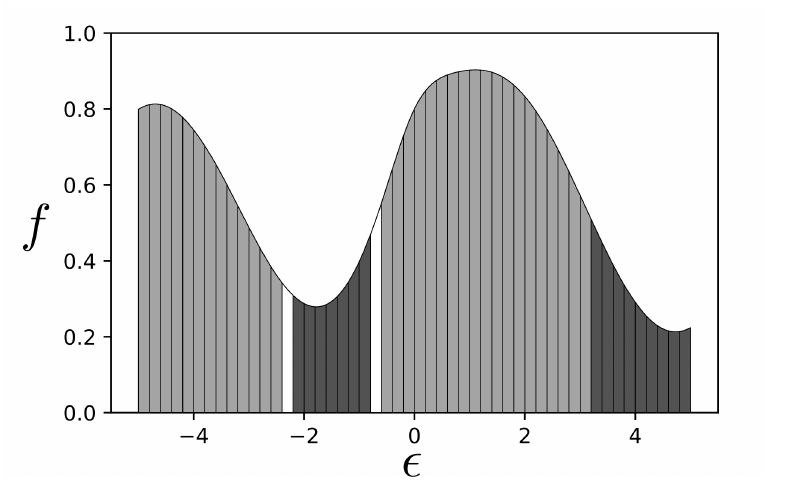}
    \quad
    \includegraphics[width=0.48\textwidth]{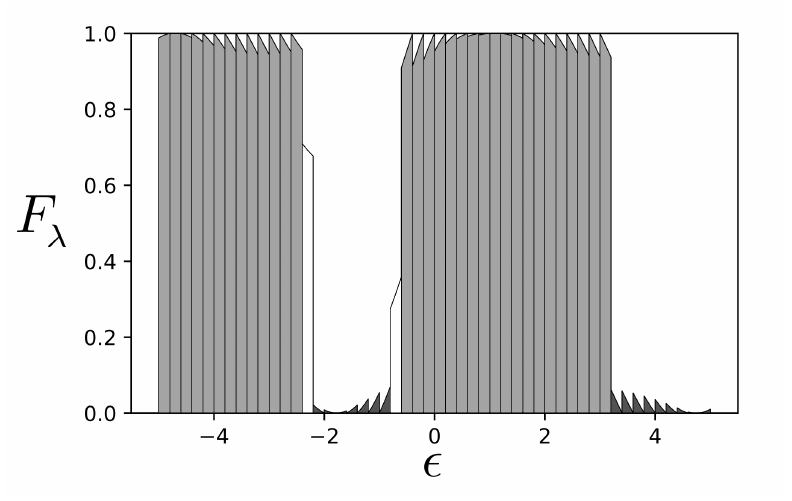}
    \caption{\fix{}{Starting with the function $f$ and the partition in figure \ref{fig:partition1}, we consider a hypothetical solution for a problem with 2 constraints. In light gray, we have the intervals corresponding to $\lambda_i = 1$, in dark gray, the intervals corresponding to $\lambda_i = 0$ and in white for $\lambda_i \neq 0, 1$ (there are at most 2 such intervals, since there are only 2 constraints). To the left we have the initial function $f$ and to the right, the optimized function $F_{\lambda}$ that is clearly closer to a boxcar collection.}}
    \label{fig:partition2}
\end{figure}

The second step in the proof is that this implies that $\mathsf{Q} = \mathsf{L}$. Let $\delta, \varepsilon > 0$ and let $u\in\mathcal{F}$ be such that $\mathcal{L}[u] \leq \mathsf{L} + \delta$. By what we just proved, there exists $U\in\mathcal{F}$ such that $\mathcal{L}[U] \leq \mathcal{L}[u] \leq \mathsf{L} + \delta$ and $\mathcal{B}[U] \leq \varepsilon$. We now call attention to the fact that $\mathcal{Q} = \mathcal{L} + \mathcal{B}$, so we have

\[
\mathsf{Q} \leq \mathcal{Q}[U] = \mathcal{L}[U] + \mathcal{B}[U] \leq \mathsf{L} + \delta + \varepsilon.
\]
Since this is true for any $\delta, \varepsilon > 0$ it implies $\mathsf{Q} \leq \mathsf{L}$. On the other hand, if $v\in\mathcal{F}$ is such that $\mathcal{Q}[v] \leq \mathsf{Q} + \delta$, then there exists $V\in\mathcal{F}$ such that $\mathcal{Q}[V] \leq \mathcal{Q}[v] \leq \mathsf{Q} + \delta$ and $\mathcal{B}[V] \leq \varepsilon$. So

\[
\mathsf{L} \leq \mathcal{L}[V] = \mathcal{Q}[V] - \mathcal{B}[V] \leq \mathsf{Q} + \delta
\]
because $\mathcal{B}[V] \geq 0$. Once again, since this is true for any $\delta, \varepsilon > 0$ it implies $\mathsf{L} \leq \mathsf{Q}$ and hence that $\mathsf{L} = \mathsf{Q}$.

This leads us to the last step of the proof. Let $\phi\in\mathcal{F}$ be such that $\mathcal{Q}[\phi] \leq \mathsf{Q} + \delta$, then we have

\[
\mathcal{L}[\phi] + \mathcal{B}[\phi] \leq \mathsf{Q} + \delta \Rightarrow \mathcal{B}[\phi] \leq \delta + \mathsf{Q} - \mathcal{L}[\phi] = \delta + \mathsf{L} - \mathcal{L}[\phi] \leq \delta .
\]
\end{proof}

What we proved here is an actual proof of the statement in the main text, because it means that as we consider transmission functions that are closer and closer to being optimal ($\mathcal{Q}[\mathcal{T}]\rightarrow \mathsf{Q}$) we will be taking $\delta \rightarrow 0$ and hence these transmission functions are getting closer and closer to being a boxcar ($\mathcal{B}[\mathcal{T}]\rightarrow 0$). In particular, if $\mathcal{T}$ is an optimal function, then $\mathcal{Q}[\mathcal{T}] = \mathsf{Q}$ implying that $\mathcal{Q}[\mathcal{T}] \leq \mathsf{Q} + \delta$ for all $\delta \geq 0$, so $\mathcal{B}[\mathcal{T}]\leq \delta$ in this case would imply that $\mathcal{B}[\mathcal{T}] = 0$ and hence that $\mathcal{T}$ is a boxcar almost everywhere.

\subsection{\label{supmat:theorem2}Proof of theorem 2}

We now move our attention to the second theorem stated in the main text, namely

\begin{theorem}
The position of the optimal boxcars is determined by the regions in energy for which 
\[
g(\epsilon) \leqslant (\lambda \epsilon + \eta) \Delta f(\epsilon),
\]
where $\lambda$ and $\eta$ are Lagrange multipliers introduced to fix $I$ and $J$. 
\end{theorem}

\begin{proof}
Using theorem 1, it follows from Eqs.~(2) and (3) of the main text that, for a transmission function $\mathcal{T}$ that is optimal we must have
\[
I(a_1, b_1,\ldots, b_N) = \sum_i\int_{a_i}^{b_i} \Delta f(\epsilon) d\epsilon,\quad\quad J(a_1, b_1,\ldots, b_N) = \sum_i\int_{a_i}^{b_i} \Delta f(\epsilon) \epsilon d\epsilon \quad\quad\mbox{and}
\]\[
\Delta^2_{I,\mathrm{opt}}(a_1, b_1,\ldots, b_N) = \sum_i\int_{a_i}^{b_i} g(\epsilon) d\epsilon,
\]
where, recall, $\Delta f(\epsilon) = f_L(\epsilon) - f_R(\epsilon)$ and $g(\epsilon) = f_L(1-f_L) + f_R(1-f_R)$.
So if we consider a Lagrangian $L(a_1, b_1,\ldots, b_N;\lambda, \eta) = \Delta^2_{I,\mathrm{opt}} + \lambda(J-J_0) + \eta(I-I_0)$, then the extrema are solutions of the system

\[
\left\{
\begin{array}{l}
g(a_i) = (\lambda a_i + \eta) \Delta f(a_i) \\
g(b_i) = (\lambda b_i + \eta) \Delta f(b_i) \\
I(a_1, b_1,\ldots, b_N) = I_0 \\
J(a_1, b_1,\ldots, b_N) = J_0 \\
a_1 \leq b_1 \leq a_2 \leq b_2 \leq \ldots \leq a_N \leq b_N
\end{array}
\right.
\]

As such we can further determine that the endpoints of the boxcar predicted in theorem 1 obey $g(\epsilon) = (\lambda \epsilon + \eta) \Delta f(\epsilon)$. However, this still doesn't tell us what are the intervals that constitute the actual optimal boxcar. To find this out, let us consider first the values of the multipliers ($\lambda^*, \eta^*$) for which the optimal boxcar is attained and let $\Pi$ be the partition of the line created by the solutions in $\epsilon$ to $g(\epsilon) = (\lambda^* \epsilon + \eta^*) \Delta f(\epsilon)$. Exploring the fact that $\mathsf{Q} = \mathsf{L}$ (as seen in the proof for theorem 1), we can then consider the problem of minimizing $\mathcal{L}$ for a transmission function in $\mathcal{R}$ that is constant on the intervals in $\Pi$ (so that it is piecewise constant) and satisfies the constraints for $I$ and $J$. This optimization problem would read

\begin{framed}
Minimize
\[
\sum_{i=1}^M A_i \tau_i \mbox{ subject to the constraints}
\]
\begin{equation}
0\leq \tau_i \leq 1,\quad\quad\quad \sum_{i=1}^M B_i \tau_i = I_0\quad\quad\quad\mbox{and}\quad\quad\quad \sum_{i=1}^M C_i \tau_i = J_0,
\label{KKT-programm}
\end{equation}
\end{framed}
where
\[
A_i = \int_{\mathcal{I}_i} g(\epsilon) d\epsilon ,\quad\quad\quad B_i = \int_{\mathcal{I}_i} \Delta f(\epsilon) d\epsilon, \quad\quad\quad C_i = \int_{\mathcal{I}_i} \epsilon \Delta f(\epsilon) d\epsilon .
\]
and $\mathcal{I}_1, \ldots, \mathcal{I}_M$ are the intervals in $\Pi$. By construction there is an optimal solution where the $\tau_i$ (the value of the transmission function in $\mathcal{I}_i$) are either 0 or 1. Examining the Karush–Kuhn–Tucker (KKT) conditions \cite{lin-prog, vanderberghe} of the problem (\ref{KKT-programm}) we get

\begin{framed}
\begin{equation}
x_i, y_i \geq 0 ,\quad\quad x_i (\tau_i - 1) = y_i \tau_i = 0 \quad\quad\mbox{and}\quad\quad A_i + (x_i - y_i) - \lambda^* C_i - \eta^* B_i = 0,\quad\quad\mbox{for }i = 1, \ldots, M .
\label{KKT-conditions}
\end{equation}
\end{framed}
\noindent \fix{}{where the $x_i$ and $y_i$ are auxiliary (slack) variables.}

It is easy to show that a consequence of (\ref{KKT-conditions}) is that 

\[
\tau_i = 0 \Rightarrow A_i - \lambda^* C_i - \eta^* B_i \geq 0 \quad\quad\mbox{and}\quad\quad \tau_i = 1 \Rightarrow A_i - \lambda^* C_i - \eta^* B_i \leq 0 .
\]
Noticing that 
\[
A_i - \lambda^* C_i - \eta^* B_i = \int_{\mathcal{I}_i} g(\epsilon) - (\lambda^* \epsilon + \eta^*) \Delta f(\epsilon) d\epsilon ,
\]
and that by hypothesis $g(\epsilon) - (\lambda^* \epsilon + \eta^*) \Delta f(\epsilon)$ doesn't change signs inside of $\mathcal{I}_i$, this implies that the optimal boxcar (which is the one satisfying the KKT conditions) must be the one determined by: 

\[
g(\epsilon) - (\lambda^* \epsilon + \eta^*) \Delta f(\epsilon) \leq 0 ,
\]
concluding the proof.
\end{proof}

\section{Additional remarks}
\label{app:misc}

\fix{}{
\subsection{Changes in the boxcar configurations}

As was seen in the main text, a variety of different boxcar configurations are possible (see for example figure \ref{fig:example1}). Since the endpoints of the boxcars are given by the solutions of $g(x) = (\lambda x + \eta) \Delta f (x)$, then the implicit function theorem implies that these endpoints can only change continuously, except for the cases where the endpoint corresponds to a solution with multiplicity higher than 1 or to a solution coming from $\pm \infty$. What this means is that these 2 extreme cases correspond to the situations where the configuration can change.

More concretely, $z$ is a solution of $g(x) = (\lambda x + \eta) \Delta f (x)$ with multiplicity higher than 1 iff $z$ is a root of $\mathcal{G}(x) - (\lambda x + \eta)$ with multiplicity higher than 1 (where $\mathcal{G} = \nicefrac{g}{\Delta f}$), which is equivalent to $\lambda = \mathcal{G}'(z)$, $\eta = \mathcal{G}(z) - z\mathcal{G}'(z)$. This corresponds to the situations where boxcars merge or split, as well as the situations where boxcars pop in and out of existence at finite energies (see figures \ref{fig:bif-1} and \ref{fig:bif-2}).

\begin{figure}[htbp!]
    \centering
    \includegraphics[width=0.6\textwidth]{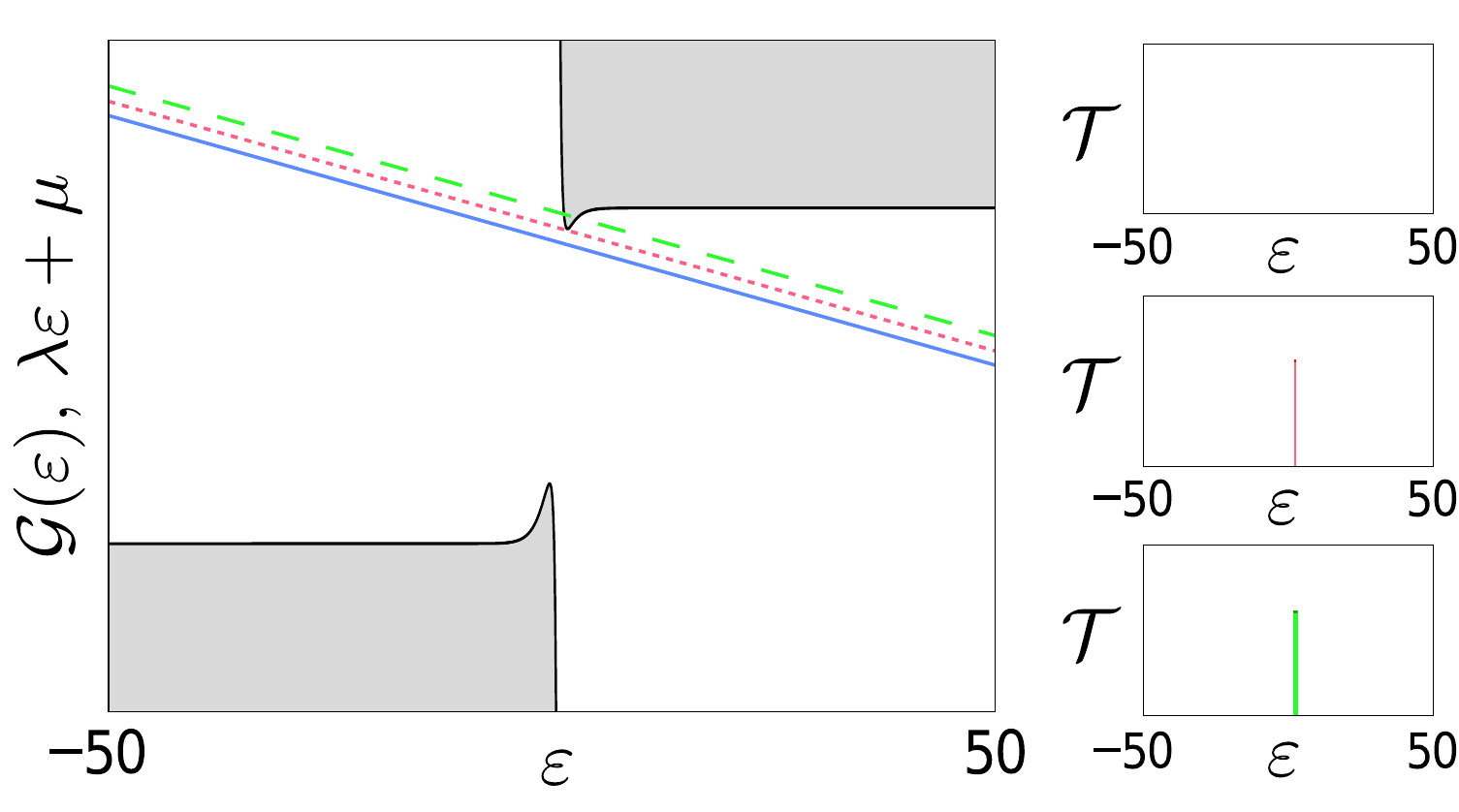}
    \caption{\fix{}{A transition between 2 configurations where a boxcar pops into existence. The bottom blue line does not intersect the gray region and hence corresponds to an empty boxcar (a possible solution for $I=0$, $J=0$) as seen in the top-right graph. Increasing $\eta$ until the line is tangent (red dotted) leads to a single point of intersection and a boxcar concentrated in a single energy as seen in the center-right graph. Further increasing $\eta$ (green dashed) the intersection increases and the corresponding boxcar widens (bottom-right graph).}}
    \label{fig:bif-1}
\end{figure}

\begin{figure}[htbp!]
    \centering
    \includegraphics[width=0.6\textwidth]{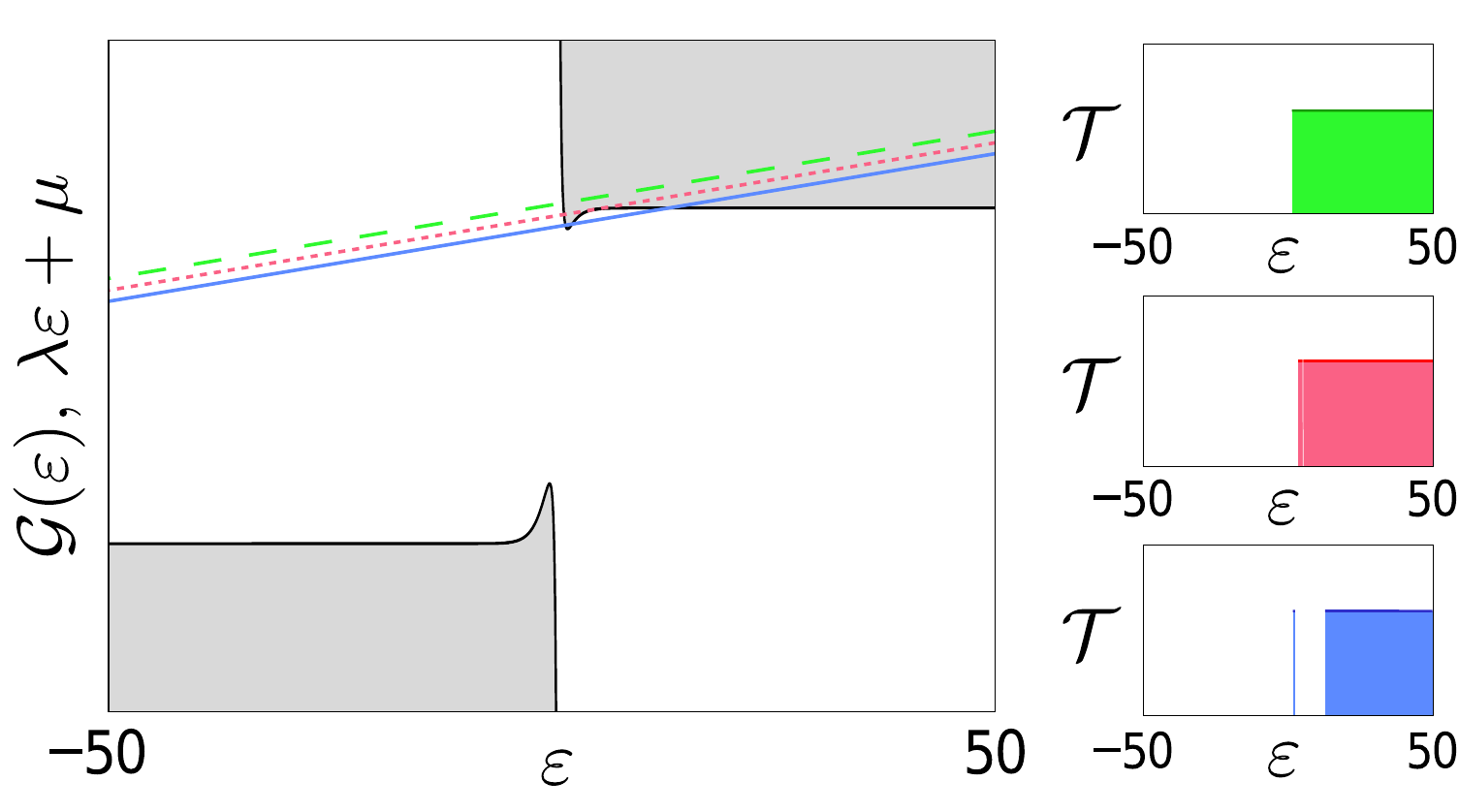}
    \caption{\fix{}{A transition between 2 configurations where two boxcars merge. The bottom blue line has as intersection with the gray area that has 2 components (plus a third one for $\epsilon \ll 0$ that we'll ignore). This leads to a configuration that has a gap in the transmission function (top-right graph). If we increase $\eta$ (red dotted) the intersection widens until the 2 components become arbitrarily close and the gap arbitrarily narrow (center-right graph). Further increasing $\eta$ (green dashed) leads to the components becoming one and the gap disappearing (bottom-right graph).}}
    \label{fig:bif-2}
\end{figure}

The situations where solutions emerge from $\pm \infty$ happen when these solutions emerge for the equation $\mathcal{G}(x) = \lambda x + \eta$. This happens when $\lambda = 0$ as a consequence of the fact that $\mathcal{G}$ converges to a finite value at $x = \pm \infty$ (this is illustrated in figure \ref{fig:bif-3}).

\begin{figure}[htbp!]
    \centering
    \includegraphics[width=0.6\textwidth]{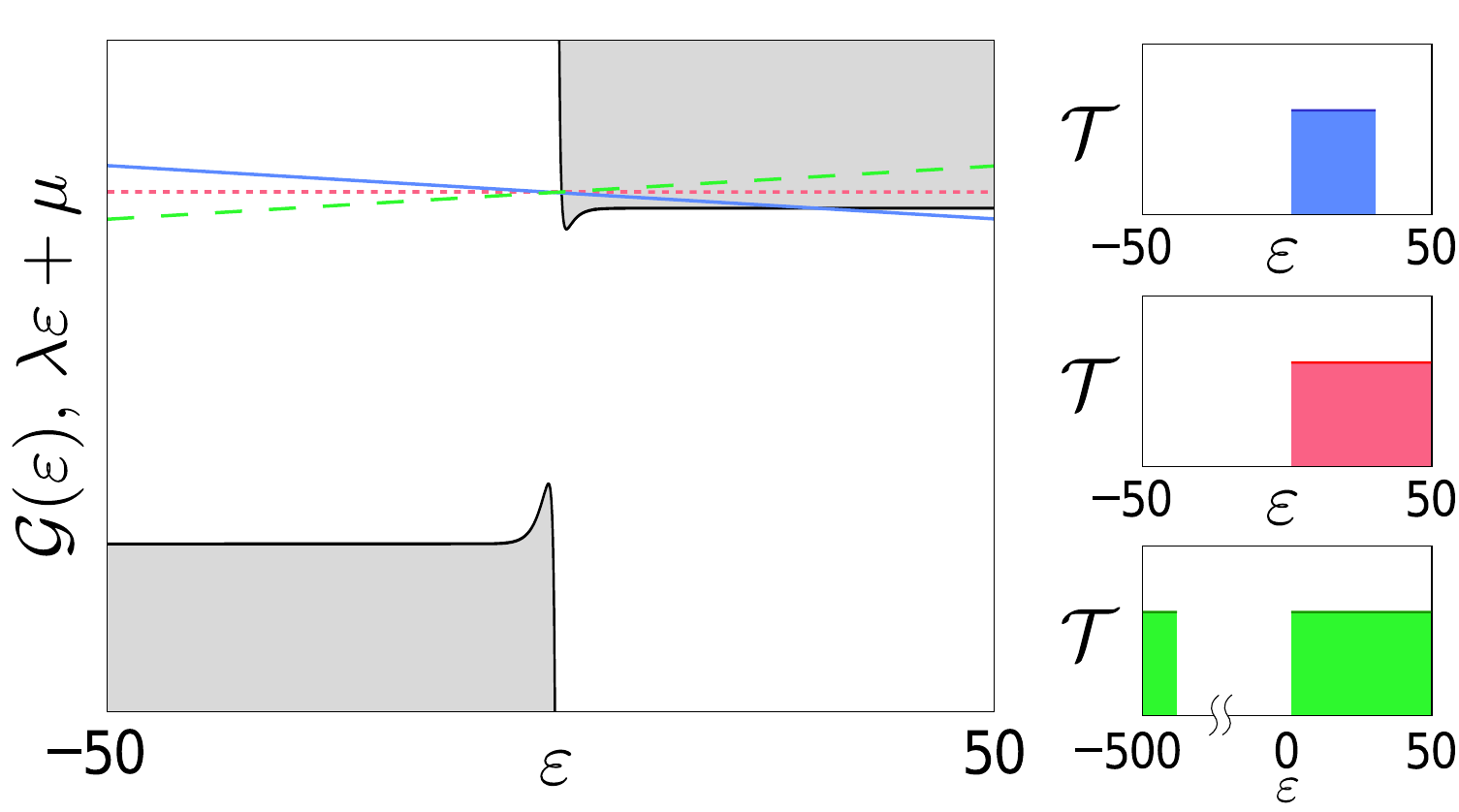}
    \caption{\fix{}{A transition between 2 configurations where solutions at $\pm\infty$ emerge. For the fully horizontal line (red dotted) the intersection with the gray region extends all the way to $\infty$ and there is no intersection for $\epsilon < 0$, leading to an infinite boxcar (center-right graph). However, if we decrease $\lambda$ from 0 to a negative value (blue) the intersection no longer extends all the way to $\infty$, so now we have a finite boxcar (top-right graph). On the other hand, increasing $\lambda$ to a positive value (green dashed; keeps the intersection going all the way to $\infty$ but also creates a new one for $\epsilon \ll 0$ (happening beyond the scale of the graph to the left). This leads to a situation with 2 boxcars each extending to $\infty$ and $-\infty$ respectively (bottom-right graph - to display the second boxcar we break the scale in $\epsilon$).}}
    \label{fig:bif-3}
\end{figure}



}

\subsection{Derivatives and monotonicity of $I(\lambda, \eta)$ and $J(\lambda, \eta)$}
\label{supmat:monotonic}

\fix{When $(\lambda, \eta) \notin\mathcal{B}$, then t}{T}he functions $I$ and $J$ can be differentiated via the implicit function theorem\fix{.}{}~\fix{In this case,}{when} $R(x; \lambda, \eta) \equiv g(x) - (\lambda x + \eta) \Delta f(x)$ only has simple roots. Let $x^*$ be such a root. Applying the implicit function theorem yields

\begin{equation}
\frac{\partial x^*}{\partial \eta} = \frac{\Delta f(x^*)}{R'(x^*;\lambda, \eta)}\quad\quad\mbox{and}\quad\quad \frac{\partial x^*}{\partial \lambda} = \frac{x^*\Delta f(x^*)}{R'(x^*;\lambda, \eta)}.
\label{eq:TFI}
\end{equation}

The boxcar corresponding to $(\lambda, \eta)$ is the indicator function of an union $\bigcup_{k} [a_k,b_k]$, with $a_1 < b_1 < \ldots < a_{n} < b_n$ (where we can potentially have $a_1 = -\infty$ or $b_n = \infty$). Define then
\[
\xi_k = 
\left\{
\begin{array}{cl}
0 & \mbox{if } a_k = -\infty \\
\frac{\Delta f(a_k)^2}{R'(a_k;\lambda, \eta)} & \mbox{otherwise}
\end{array}
\right.
\quad\quad\mbox{and}\quad\quad
\theta_k = 
\left\{
\begin{array}{cl}
0 & \mbox{if } b_k = \infty \\
\frac{\Delta f(b_k)^2}{R'(b_k;\lambda, \eta)} & \mbox{otherwise}
\end{array} .
\right.
\]
Since $I$ and $J$ will be given by

\[
I = \sum_i\int_{a_i}^{b_i} \Delta f(\epsilon) d\epsilon,\quad\quad J = \sum_i\int_{a_i}^{b_i} \Delta f(\epsilon) \epsilon d\epsilon ,
\]
then from eq (\ref{eq:TFI}) one obtains

\[
\frac{\partial I}{\partial \eta} = \sum_k (\theta_k - \xi_k), \quad\quad
\frac{\partial I}{\partial \lambda} = \frac{\partial J}{\partial \eta} = \sum_k (\theta_k b_k - \xi_k a_k)\quad\quad\mbox{and}\quad\quad \frac{\partial J}{\partial \lambda} = \sum_k (\theta_k b_k^2 - \xi_k a_k^2) .
\]

Since $a_k$ is always the leftmost end of one of the boxcar pieces, then either $a_k = -\infty$ or $R'(a_k;\lambda, \eta) < 0$. Similarly, since the $b_k$ are rightmost ends, then either $b_k = \infty$ or $R'(b_k;\lambda, \eta) > 0$. From their definition, this implies that either $a_k = -\infty$ or $\xi_k < 0$ and either $b_k = \infty$ or $\theta_k > 0$ (note that this is a consequence of~\fix{$\varepsilon_S$ never being an endpoint}{theorem 2}). So comparing with the expressions obtained, we see that if $(\lambda, \eta)$ doesn't correspond to $\mathcal{T} = 0$ or $\mathcal{T} = 1$ everywhere (that is, $R$ has at least one root), then we have

\[
\frac{\partial I}{\partial \eta} > 0\quad\quad \mbox{and}\quad\quad
\frac{\partial J}{\partial \lambda} > 0 , 
\]
otherwise both derivatives are 0.

\subsection{Continuity of $\Delta_{I,{\rm opt}}^2(I,J)$}

We can show that the optimal value\fix{ of}{} $\Delta_{I,{\rm opt}}^2$ is a continuous function~\fix{}{of the constrained currents $I$ and $J$}. This follows from showing that $\mathsf{L}$ as a function of the constraint values $\phi$ is a convex function (and that $\mathsf{L} = \mathsf{Q}$, as was shown in the proof of theorem \ref{teor:boxcar})

\begin{theorem}
$\mathsf{L}(\phi)$ is convex as a function of $\phi$.
\label{lemma:L-convex}
\end{theorem}

\begin{proof}
Let $\mathcal{F}(\phi)$ denote the feasible set for a problem with constraints $\phi$ and let $\phi_1, \phi_2$, such that $\mathcal{F}(\phi_1), \mathcal{F}(\phi_2) \neq \varnothing$ (and hence such that $\mathsf{L}$ is finite). For all $\varepsilon > 0$ we can then find functions $f_1, f_2$ such that $f_i \in \mathcal{F}(\phi_i)$ and $\mathcal{L}[f_i] \leq \mathsf{L}(\phi_i) + \varepsilon$. It follows that for every $t\in [0,1]$, we have $F_t\equiv (t f_1 + (1-t) f_2) \in \mathcal{F}(t \phi_1 + (1-t) \phi_2)$. Moreover:

\[
\mathcal{L}[f_1] \leq \mathsf{L}(\phi_1) + \varepsilon \quad\quad\mbox{and}\quad\quad \mathcal{L}[f_2] \leq \mathsf{L}(\phi_2) + \varepsilon \quad\mbox{imply that}
\]\[
\mathcal{L}[F_t] \leq t \mathsf{L}(\phi_1) + (1-t) \mathsf{L}(\phi_2) + \varepsilon \Rightarrow \mathsf{L}(t \phi_1 + (1-t) \phi_2) \leq t \mathsf{L}(\phi_1) + (1-t) \mathsf{L}(\phi_2) + \varepsilon , 
\]
and since this must hold for all~\fix{$\varepsilon \geq 0$}{$\varepsilon > 0$}, then it follows that

\[
\mathsf{L}(t \phi_1 + (1-t) \phi_2) \leq t \mathsf{L}(\phi_1) + (1-t) \mathsf{L}(\phi_2) ,
\]
implying convexity.
\end{proof}

\bibliography{library_OK}
\end{document}